\pdfoutput=1

\documentclass[reqno]{amsart}
\usepackage[foot]{amsaddr}
\usepackage{amssymb}
\usepackage[hyperindex,breaklinks]{hyperref}
\usepackage{caption}
\usepackage{eucal}
\usepackage{empheq}
\usepackage[T2A,OT1]{fontenc}

\usepackage{csquotes}
\usepackage[russian,french,english]{babel}

\usepackage{epigraph}
\usepackage{bm}
\usepackage{dsfont}
\usepackage{mathtools}
\usepackage{booktabs}
\usepackage{multirow}
\usepackage{tikz}
\usepackage[sorting=none,
autolang=other]{biblatex}
\DeclareFieldFormat{eprint:CiteSeerX}{%
  CiteSeerX: \href{https://citeseerx.ist.psu.edu/viewdoc/summary?doi=#1}{#1}%
}

\usepackage{wrapfig}
\usepackage{enumitem}
\setlist[itemize]{itemindent=2em, leftmargin=0pt}
\setlist[enumerate]{label = \textup{(\roman*)}}
\setlist[description]{itemindent=2em, leftmargin=0pt}

\usepackage{makecell}

\usetikzlibrary{trees}

\makeatletter
\newcommand{\lowtilde}[1]{{%
  \mathpalette\low@tilde{#1}%
}}
\newcommand{\low@tilde}[2]{%
  \sbox\z@{$\m@th#1{#2}$}%
  \ht\z@=.9\ht\z@
  {\widetilde{\box\z@}}%
}
\makeatother
\newcommand{\dtilde}[1]{\lowtilde{\lowtilde{#1}}}

\AtBeginBibliography{\small}

\theoremstyle{definition}\newtheorem{Def}{Definition}[section]
\theoremstyle{definition}\newtheorem{Asmp}{Assumption}[section]
\theoremstyle{definition}\newtheorem{Ex}{Example}[section]
\theoremstyle{remark}\newtheorem{Inf}{Inference}
\theoremstyle{remark}\newtheorem{Con}{Conclusion}
\theoremstyle{remark}
\theoremstyle{definition}\newtheorem{Conj}{Conjecture}[section]
\theoremstyle{plain}\newtheorem{Th}{Theorem}[section]
\theoremstyle{plain}\newtheorem{Fact}{Fact}[section]
\theoremstyle{plain}\newtheorem{Le}{Proposition}[section]

\newcommand{\cii}[1]{_{{}_{#1}}}

\newcommand{\df}{\buildrel{}_\mathrm{def}\over=}
\newcommand{\assign}{\mathbin{@}}

\DeclareMathOperator{\mean}{Mean}
\DeclareMathOperator{\med}{Median}
\DeclareMathOperator{\lmd}{lmd}
\DeclareMathOperator{\rhofn}{rho}
\DeclareMathOperator{\logit}{logit}
\DeclareMathOperator{\expit}{expit}

\DeclareMathOperator*{\argmax}{arg\,max}

\author[Alexey V. Osipov]{Alexey V. Osipov$^{\ast,\dagger}$}
\address{$^\ast$St. Petersburg State University, 
St. Petersburg, Russia}
\address{$^\dagger$Syncretis LLC, St. Petersburg, Russia}
\author[Nikolay N. Osipov]{Nikolay N. Osipov$^{\ast,
\ddagger,\mathsection}$}
\address{$^\ddagger$St. Petersburg Department of V.~A.~Steklov Institute of Mathematics of the Russian Academy of Sciences, St. Petersburg, Russia}
\address{$^\mathsection$HSE University, International Laboratory of Game Theory and Decision Making, St. Petersburg, Russia}
\email{nicknick AT pdmi DOT ras DOT ru}

\thanks{
Both authors were supported by the Chebyshev Laboratory under RF Government grant 11.G34.31.0026. The second author was also supported by ERCIM ``Alain Bensoussan'' Fellowship, by the HSE University Basic Research Program, and by the Foundation for the Advancement of Theoretical Physics and Mathematics ``BASIS''. 
A major revision 
was carried out with the support of a grant of the Government of the Russian Federation for the state support of scientific research, carried out under the supervision of leading scientists, agreement 075-15-2021-602}

\hypersetup{
pdfstartview = FitH,colorlinks = true, linkcolor = blue, citecolor = red}

\title[Interpretable collective intelligence]{From prediction markets to interpretable collective intelligence}

\begin{document}
\begin{abstract}
We outline how to create a mechanism that provides an optimal way to elicit, from 
an arbitrary group of experts, the probability of the truth of an arbitrary logical proposition together with collective information that has an explicit form and interprets this probability. Namely, we provide strong arguments for the possibility of the development of a self-resolving prediction market with play money that incentivizes
direct information exchange between experts. Such a system could, in particular, 
motivate simultaneously many experts 
to collectively solve scientific or medical problems in a very efficient manner. We also note that in our considerations, experts are not assumed to be Bayesian.
\end{abstract}

\maketitle
\setcounter{secnumdepth}{2}


\section{Introduction}
\label{subsec:goal} 
Suppose we have a proposition that may be either true or false.  
Our goal is to provide strong arguments for the possibility of the creation of a system that can elicit, from 
an arbitrary group of experts, the probability of the truth of the proposition together with
the relevant collective information in a comprehensible form. This information 
will fully substantiate (interpret) the group probability. More precisely, we outline a way to create a \emph{self-resolving} prediction market with 
\emph{play money} that, in addition, incentivizes \emph{direct} information exchange between experts. 
This simultaneously solves the following three problems:
\begin{itemize}
    \item information exchange in an ordinary prediction market may be broken;
    \item all information aggregated by an ordinary market is encoded in the final price, which is a non-interpretable value between $0$ and $1$;
    \item we can resolve an ordinary market and distribute gains and losses \emph{only} after it becomes publicly known whether the proposition in question is true or false.
\end{itemize}
Since the proposed system overcomes these limitations, it could 
incentivize experts from all over the world to brainstorm scientific or medical problems in 
a very efficient manner.

We say that a prediction market is self-resolving if it distributes final gains and losses relying on its internal state and does not need verifiable ground truth for that. We borrow this term from \cite{AhWi2018,Ah2019}, where similar markets are experimentally studied in another context. Another example of self-resolving markets is described in \cite{DaKiLo2011}. Different approaches to the situation where ground truth is non-verifiable, are developed, for example, in~\cite{Pr2004} or in~\cite{Ba2018}. 
We note that in addition to the admissibility of non-verifiable ground truth, our proposed system could enhance the dynamic mutual learning of experts through price (a core property of prediction markets) into a direct exchange of knowledge of arbitrary nature.

Note that even conventional prediction markets have proved to be useful in science and medicine, for example, for making diagnosis in complex medical cases \cite{MeLoSi2016}, for forecasting the spread of infectious diseases \cite{PoNeNe2007}, or even for solving some general scientific problems \cite{AlKiPf2009}. We could meet the same challenges with our system, while taking full advantage of the self-resolvability and interpretability properties. 
In fact, these properties remove any restrictions on the range of problems that can be addressed with the system.



We could even try the following general approach. 
For any problem, we could put into the system a logical proposition that the problem will be solved after collecting information in the system (surely, we need to describe exactly what we mean by ``will be solved''). 
The system will incentivize experts to share 
reasoned arguments about what obstacles might be in solving the problem, or what might help to overcome known difficulties and to solve it anyway. 
The final result will be an aggregated information concerning 
the problem, together with the probability that it will be solved, given this information.

Later, we will provide some preliminaries on regular prediction markets. If the reader is already familiar with this concept, the following preliminary intuition about the system we propose may be also helpful. 
 \begin{itemize}
     \item We consider a binary prediction market and offer experts to trade as if it were 
     related to a certain proposition~$\bm\omega$ that may be either true or false.
     \item 
     At the same time, we make it public that final gains and losses will actually be distributed
     in accordance with a binary random generator 
     with the probability parameter equal to the final price in the market.
     \item We assume that the final price cannot be distorted by trading actions of any individual expert, but can be changed due to the actions of many experts. We will discuss how to achieve such a property in practice.
     \item One of the possible strategies for an expert is the following. If an expert sees that his or her individual probability of truth of~$\bm\omega$ does not coincide with the current price, then he or she can make
     a transaction as if the market were resolved in accordance with~$\bm\omega$, and after that share his or her information that explains why the current price is not fair with respect to the probability of~$\bm\omega$. 
     The purpose of information sharing is to change others' beliefs on~$\bm\omega$ and, as a consequence, to favorably alter the final price, which will be the random generator parameter. 
     \item It can be seen that such strategies form a natural Nash equilibrium. In particular, this fact leads experts to follow our initial suggestion to behave as if the market were regular and related to~$\bm\omega$. 
 \end{itemize}
In Section~\ref{sec:CI}, we will provide a formal theoretical model that incorporates this intuition, and discuss its implementability.  
We note in advance that our models do not assume that experts are Bayesian, i.e. their beliefs do not need to be conditional probabilities.

\section{Preliminaries}
First, we present our basic notation, which is partially borrowed from~\cite{Wi2009} and is convenient for our non-Bayesian setting.
\begin{itemize}
\item
Suppose there is a proposition that 
may be revealed to be either true or false. 
This means that we have a propositional variable~$\bm\omega$. 
    By $\bm E \df \{1,\dots,N\}$ we denote a population of experts whose collective belief and information about~$\bm\omega$ we would like to obtain. Our object of study is the pair~${(\bm\omega, \bm E)}$.
    \item The notions of the truth and falsity of propositions can be formalized in terms of assignments to propositional variables. Let $A$ be a set of propositional variables. Then the notation $\alpha\assign A$ means that $\alpha$ is 
    an assignment to $A$, i.e. $\alpha$ is a function from $A$ to $\Omega\df\{0,1\}$. An assignment is also called a world. 
    In particular, 
    for the one-element set~$\{\bm\omega\}$ we have two assignments: the world $\bm\omega \mapsto 1$ where $\bm\omega$ is true, and the world $\bm\omega \mapsto 0$ where $\bm\omega$ is false. 
    \item We tell that worlds $\alpha\assign A$ and $\beta\assign B$ are consistent (and write $\alpha\sim\beta$) if 
    they assigns the same values to the variables in $A\cap B$: $\alpha\vert\cii{A\cap B} = \beta\vert\cii{A\cap B}$. 
\end{itemize}

\subsection{Experts' probabilities}\label{subsec:prob}
\begin{Asmp}\label{asmp:exp_prob}
    There exists a finite set $\bm\Gamma$ such that for each expert $n \in \bm E$, his or her probability
    of $\bm\omega \mapsto 1$ is a number $\pi_n(\bm\omega)\in[0,1]$ that 
    is completely determined by a certain subset 
    $\Gamma_n \subseteq \bm\Gamma$, i.e. ${\pi_n(\bm\omega) = \bm\pi(\bm\omega\mid\Gamma_n)}$, 
    where $\bm\pi(\bm\omega\mid\cdot)$ is a fixed set function from
    $2^{\bm\Gamma}$ to $[0,1]$. We call the elements of~$\bm \Gamma$ signals and call each set~$\Gamma_n$
    the knowledge of the expert~$n$. The experts' probabilities of $\bm\omega \mapsto 0$ are determined in the same way by the set function $1 - \bm\pi(\bm\omega\mid\cdot)$. We also assume that 
    ${\bm\Gamma = \bigcup_{n\in \bm E}\, \Gamma_n}$.
\end{Asmp}

In some cases it will be convenient for us to rely on the additional assumption that the elements of $\bm \Gamma$ are propositional variables with certain values already assigned to them: there is collective information $\bm\gamma\cii{\bm{E}}\assign\bm\Gamma$ generated by the real world, and each expert $n\in\bm E$ knows the corresponding part $\gamma_n\assign\Gamma_n$ of it determined by the relation $\gamma_n \sim \bm\gamma\cii{\bm{E}}$ 
(or, what is the same, by the identity $\gamma_n = \bm\gamma\cii{\bm{E}}\vert\cii{\Gamma_n}$). 

We may interpret $\bm\Gamma$ as a set containing at least \emph{every} signal that is 
\begin{itemize}
\item
relevant to~$\bm\omega$, 
\item
is known to at~least one of the experts, 
\item
and is not known to all the experts at once.
\end{itemize}
In other words, there are only irrelevant or public signals outside of~$\bm\Gamma$. This somewhat vague interpretation becomes completely formalized if we consider Assumption~\ref{asmp:exp_prob} as a consequence of basic axioms that can be found in Appendix~\ref{app:A}: see 
Assumptions~\ref{asmp:A0}--\ref{A1_1} together with Proposition~\ref{lem:irr}.

We emphasize that the functional $\bm\pi(\bm\omega\mid\cdot)$ is \emph{not} a conditional probability.
However it may (but does not have to) be generated by a Bayesian model. 
Indeed, let $\bm\Delta \df \{\bm\omega\} \cup \bm\Gamma$. Consider the set $\mathcal{A}$ of all possible 
assignments $\delta\assign \bm\Delta$ as a sample space. 
Let $P$ be a probability measure over $2^\mathcal{A}$.
Any assignment $\delta'\assign\Delta'$ to any subset $\Delta'\subseteq\bm\Delta$ can be associated 
with the event $\{\delta \in \mathcal{A}\mid \delta\sim \delta'\} \in 2^\mathcal{A}$. 
This applies to
the assignment $\bm \omega \mapsto 1$ as well as to any assignment $\gamma' \assign \Gamma'$ for any $\Gamma'\subseteq\bm\Gamma$ 
(in particular, to $\bm\gamma\cii{\bm{E}}\vert\cii{\Gamma'}$ and to $\bm\gamma\cii{\bm{E}}$ itself). 
Relying on such associating and additionally assuming that $P(\bm\gamma\cii{\bm{E}}) \ne 0$, 
we \emph{could} introduce our functional by the formula
\begin{equation}\label{eq:B_prob}
\bm\pi(\bm\omega\mid\Gamma') 
= P\bigl(\bm\omega \mapsto 1 \;\bigm|\; \bm\gamma\cii{\bm{E}}\vert\cii{\Gamma'}\bigr)
\quad\mbox{for any}\quad
\Gamma'\subseteq\bm\Gamma.
\end{equation}

\subsection{Prediction markets}\label{subsec:pr_mr}
Suppose each expert $n\in\bm E$ possesses 
a certain amount of units that are valuable for him or her (real money or some play money). We use the following terminology (see, e.g.,~\cite{Ma2006}).
\begin{Def} 
\emph{A contract} on $\bm\omega\mapsto 1$ is the obligation that can be sold by one expert to another and that binds a seller to pay 
a buyer $1$~unit once it is revealed that $\bm\omega$ is true.
\end{Def}

We can associate with $\bm\omega$ a prediction market where the experts would be trading contracts on 
$\bm\omega\mapsto 1$ with each other. 
The potential ability of the price in a prediction market to aggregate experts' information and to converge to a collective probability of $\bm\omega\mapsto 1$ has been the subject of many empirical and theoretical studies.
Below we provide a brief description of a classical Bayesian model, and many further details can be found in Appendix~\ref{app:pr_mr}. 

But before the formal model, we \emph{in}formally describe an important example that is a slight variation of an example from \cite{PlSu1988} and shows how experts' information and beliefs may vary during the trade even if the initial beliefs are identical. 
Suppose there is an event with three mutually exclusive outcomes $x$, $y$, and $z$. 
Let $\bm\omega$ be the proposition that $x$ has been realized. Suppose a random half of the experts know that $y$ has \emph{not} been realized and the other half know that $z$ has \emph{not} been realized.
All the experts also know that a random half of them know about~$y$ and the other half know about~$z$.
If the experts reckon with the principle of indifference, they will have identical initial 
probabilities $\pi_n^{0}(\bm\omega) = 1/2$, $n\in\bm E$. 
In this situation, any static probability aggregation scheme yields nothing better than $1/2$.
On the other hand, if the experts generate a prediction market and are sophisticated enough to judge about others' probabilities by watching others' behavior in this market, then 
their beliefs will change toward~$1$.  

The discrete model of prediction market from \cite{FeFoPe2005} is one way to formally describe how the price and beliefs cease to fluctuate 
in a single consensus point, in particular in the $xyz$-ex\-am\-ple. 
Here we provide those part of that model that is a direct consequence of the classical results of \cite{NiBrGe1990} that, in its turn, extend the line of research started 
in \cite{Au1976} (see also Appendix~\ref{subsec:pen_res} 
for further discussion of that model). 

First, assume that we are in the Bayesian setting as it is described at the end of Section~\ref{subsec:prob}: there is a probability measure~$P$ over $2^{\mathcal{A}}$, where $\mathcal{A}$ is the set of all assignments to $\bm\Delta = \{\bm\omega\} \cup \bm\Gamma$. 
Additionally assume that the sets $\Gamma_n$ are pairwise disjoint and that each of them consists of a single signal. Of course, in this situation we can use a more conventional notation. 
Namely, following \cite{FeFoPe2005} (see also the survey~\cite{PeSa2007}), we consider the information structure 
$$\Omega\times S = \Omega \times S_1 \times \cdots \times S_N,\quad S_n \df \{0,1\},$$ and bit 
arrays $(\omega,s) = (\omega, s_1,\dots,s_N)\in \Omega\times S$. This is the same as to consider 
a set~$\bm\Delta$ consisting of
$N+1$ propositional variables, with all possible assignments~$\delta \assign \bm\Delta$.
\begin{Asmp}[the model of 
Feigenbaum et al.]
\label{asmp:FFPS}\leavevmode
\begin{enumerate}
\item\label{it:FFPS_first}
A vector $r = (r_1,\dots,r_N)\in S$ is the real state of the world \textup{(}$\bm\gamma\cii{\bm{E}}$ in our notation\textup{)}, and each expert $n\in\bm E = \{1,\dots,N\}$ privy to one bit $r_n \in S_n$ of information. 
\item\label{it:FFPS_second} The experts' probabilities of $\omega = 1$ \textup{(}of $\bm\omega \mapsto 1$ in our notation\textup{)} 
are conditional probabilities generated by a probability measure~$P$ over $2^{\Omega\times S}$ such that $P(s = r) \ne 0$ 
\textup{(}the world~$r$ is \emph{not un}expected\textup{)}.

\item\label{it:FFPS_third} 
There is a prediction market that proceeds in rounds. The clearing price is publicly revealed after each round, and it turns out to be the mean of the experts' beliefs before the round.
\item 
Such means are the only information channel between the experts. 
\item\label{it:FFPS_last} All of the above is common knowledge \textup{(}including the structure $\Omega\times S$ and the measure~$P$\textup{)}.
\end{enumerate}
\end{Asmp}

By $\pi^k_n$ and $q_k^c$ we denote the experts' beliefs before the $k$-th round and the clearing price after it, respectively. 
We can write~\ref{it:FFPS_third} as 
\begin{equation}\label{eq:q_is_mean}
    q_k^c = \mean\big\{\pi_n^k\big\}_{n\in\bm E}.
\end{equation}
We note that~\eqref{eq:q_is_mean} means that the price $q_k^c$ combines information exchange with averaging, which is more than can be said for 
the price in a real prediction market (see Appendix~\ref{sec:model}). 

By $I^k$ we denote the public information accumulated by the market for the first $k-1$ rounds:
$$
	I^1 \df \{s \in S\mid P(\Omega\times\{s\})\ne 0\},
	\quad I^{k+1} \df \{s \in I^{k}\mid q_{k}^c(s) = q_k^c \},
$$
where $q_k^c(s)$ are the prices in a world~$s$ (in particular, $q_k^c=q_k^c(r)$). 
The information of an expert~$n$ before the $k$-th round is $$I^k_n \df \{s \in I^k\mid s_n = {r}_n\},$$ and we have
$
\pi^k_n = P\bigl(\omega=1 \bigm| s \in I^{k}_n\bigr).
$
We note that events $\{(w,s)\in\Omega\times S \mid s \in I^{k}_n\}$ for $k>1$ may be irreducible to the partial assignments $\bm\gamma\cii{\bm{E}}\vert\cii{\Gamma'}$ arising in~\eqref{eq:B_prob}, and, therefore, after the first round the corresponding probabilities $\pi^k_n$ may not always be described within the non-Bayesian context.

Since $I^1 \supseteq I^2 \supseteq \cdots$, there exists a moment~$k_\infty$ when the information stops accumulating: 
\[
I^{k_\infty} = I^{k_\infty + 1} = \cdots.
\]
In \cite{FeFoPe2005}, the authors notice that the results of \cite{NiBrGe1990} immediately imply that at $k_\infty$ 
each expert's belief $\pi^{k_\infty}_n$ turns out to be independent of his or her own private information and
all the beliefs stabilize at one point. 
Namely, we have the following proposition.

\begin{Le}\label{prop:final_state}
For all $n \in \bm E$ we have 
\begin{equation}\label{eq:final_state}
	\pi^{k_\infty}_n 
	= q_{k_\infty}^c = P\big(\omega = 1 \mid s \in I^{k_\infty}\big).
\end{equation}
\end{Le}

Now we can formalize the $xyz$-ex\-am\-ple as follows.
\begin{Ex}\label{ex:xyz}
Suppose $S=S_1\times S_2$ 
and $r= (0,0)$, and set
\begin{align*}
p_x &\df P(\omega = 1,s_1 = 0,s_2 = 0),\\
p_y &\df P(\omega = 0,s_1 = 1,s_2 = 0),\\
p_z &\df P(\omega = 0,s_1 = 0,s_2 = 1). 
\end{align*}
If $p_x = p_y = p_z = \tfrac{1}{3}$, then we have $k_\infty=2$, $I^2 = \{r\}$, and $q_{2}^c = 1$.
\end{Ex}

Finally, we can state the obvious fact that nothing changes when some experts are ignorant.
\begin{Fact}\label{rem:d_less_N}
Consider a situation where the dimension~$d$ of $S$ is strictly less than~$N$: 
$d$~experts privy to $d$ bits of $r\in S$ as in~\textup{\ref{it:FFPS_first}} 
from Assumption~\textup{\ref{asmp:FFPS}}, while remaining $N-d$ experts know no bits at all. We suppose \textup{\ref{it:FFPS_second}--\ref{it:FFPS_last}} from Assumption~\textup{\ref{asmp:FFPS}} remain fulfilled for all $N$ experts: in particular, both~$N$ and~$d$ are assumed to be common knowledge. Under these conditions,
relation~\eqref{eq:final_state} remains valid and we can shorten~\eqref{eq:final_state} by dropping the second identity: it becomes a part of the first.
\end{Fact}

\section{Main results} 
\label{sec:CI}
In fact, there is empirical evidence that prediction market prices are indeed able to aggregate information, but may fail to do so if they are the only information channel between the experts: see, for example, the study \cite{AlKiPf2009}, which is also discussed in detail in Appendix~\ref{subsec:emp_evid}. 
The same study shows that direct information exchange could greatly improve the reliability of prediction markets, but does not offer a way to automatically induce such exchange.
In addition to the possibility of broken information exchange through prices, there are the following two problems. 
First, we must reliably know the value of $\bm\omega$ in order to resolve our market and to distribute gains and losses. This severely restricts the applicability of prediction markets. Second, consider an external ignorant observer who does not know the nature of experts' information (in terms of the model discussed in Section~\ref{subsec:pr_mr}, he or she 
does not know the probability space $(\Omega\times S, 2^{\Omega\times S}, P)$). Then he or she sees prices, but cannot obtain the information that generates them: results are not interpretable.
Below we present a theoretical non-Bayesian model that incorporates the idea that all these problems can (and should) be solved simultaneously. 
After that we discuss the possibility of its practical implementation. 

\subsection{Theoretical model of collective intelligence}\label{subsec:sr_theory}
First, we need a binary random generator that takes a parameter $p\in [0,1]$ and returns either~$1$ with probability~$p$ or~$0$ with probability $1-p$. Here $p$ is a probability just in the sense that it is a parameter of our random generator, which can be of a physical or algorithmic nature. Concerning algorithmic randomness, 
    see the survey \cite[Section~2.6]{UsSe1987trans} with references 
    to Knuth, Kolmogorov, Martin-L\"of, von Mises, et~al. 
    Regardless of the implementation details, we formalize the generator as the parametric family $\bm\omega[p]$ of propositional variables which correspond to the following propositions: 
    if we call the generator with parameter~$p$, then it will return~$1$. 
    We also note that there are two assignments over each one-element set $\{\bm\omega[p]\}$.



Now we are ready to present our model and its consequences. We emphasize that what follows is not the desired mechanism itself, but only its simplified model that incorporates the main properties of the system in question. The implementability of the system is discussed in Section~\ref{sec:impl}. 
\begin{Asmp}[a theoretical model of an interpretable self-resolving market]\label{asmp:SRM}\leavevmode
\begin{enumerate}
    \item\label{srm_it1} We are within Assumptions~\textup{\ref{asmp:exp_prob}}, 
    and there is a non-empty set $e_1 \subseteq \bm E$ of experts with only public or irrelevant knowledge: 
    if $n\in e_1$, then
    $$\Gamma_n = \Gamma^1 \df \varnothing.$$
        Initially, we have a prediction market where only the experts from~$e_1$ participate. 
    \item\label{srm_it3} At the beginning of each round $k$, the current participants from $e_{k}$ trade for a period of time and come to a price~$q^c_k$ known to all the experts. 
    \item\label{srm_it4} After that, a new expert $n_k\in\bm E\setminus e_k$ may join to $e_k$ 
    by making a transaction at the price~$q^c_k$: he or she may choose whether to buy or to sell contracts.
    After that, the expert may either
    share \emph{all} his or her private information with the other experts in~$\bm E$, or remain 
    completely silent (see also the discussion below).
    The next round begins with $e_{k+1} = e_k\cup\{n_k\}$, and 
    the public information
    is updated either as $\Gamma^{k+1} = \Gamma^{k} \cup \Gamma_{n_k}$ or as 
    $\Gamma^{k+1} = \Gamma^{k}$. The individual information of 
    each $n\in\bm E$ is determined as
    $\Gamma_n^k \df \Gamma^{k} \cup \Gamma_n$.
    

        
    \item Since $\bm E$ is finite, there exists a final round~$k_\infty$ 
    when no one wishes to join~$e_{k_\infty}$. After~$k_\infty$, we resolve the market \textup{(}distribute gains and losses\textup{)} in accordance with 
    $\bm\omega\big[q^c_{k_\infty}\big]$: each contract brings~$1$~unit in the world 
    ${\bm\omega\big[q^c_{k_\infty}\big]\mapsto 1}$ and nothing in the world 
    $\bm\omega\big[q^c_{k_\infty}\big]\mapsto 0$.
    If in a round~$k$ an expert $n \in \bm E$ believes that\footnote{Hereinafter, the word ``believes'' means ``takes for granted in the further logical inference''. If we would like to be even more formal, we can consider our experts as \emph{a calculus with an input} (see \cite[Section~1.3.1]{UsSe1987trans}).} he or she exactly knows the value that $q^c_{k_\infty}$ will take, then he or she treats this value as the probability of $\bm\omega\big[q^c_{k_\infty}\big]\mapsto 1$. Hereinafter, we denote such a probability as $\pi_n\big(\bm\omega\big[q^c_{k_\infty}\big]\big)$ \textup{(}see~\textup{\ref{srm_it5}} and~\textup{\ref{srm_it6})}.

    \item\label{srm_it5} All the experts are myopic and risk neutral in the following sense.
    \begin{itemize}
        \item Any expert $n\in\bm E$ 
does not consider the situation where someone outside $e_k\cup\{n\}$ participates in the market. 
If someone else joins~$e_k$, this attitude is instantly carried over to $e_{k+1}\cup\{n\}$, and so forth.
        \item 
        If the expert~$n_k$ joins to $e_k$ and has 
        a certain probability $\pi_{n_k}\big(\bm\omega\big[q^c_{k_\infty}\big]\big)$, 
        then
        his or her 
        expected profit on the joining transaction is \emph{strictly} greater than zero: 
        \begin{equation}\label{eq:ex_pr}
           \pm\Bigl(\pi_{n_k}\big(\bm\omega\big[q^c_{k_\infty}\big]\big) - q_k^c\Bigr) s^{\pm} >0.
        \end{equation}
        Here the sign~$\pm$ depends on whether $n_k$ has bought or sold contracts, and $s^\pm$ is 
        their quantity. 
        \item  If there is an expert ${n\in\bm E\setminus e_{k}}$ that can achieve \eqref{eq:ex_pr}, 
        then the round~$k$ cannot be the last.
    \end{itemize}
    
    \item\label{srm_it6} If during 
    stage~\textup{\ref{srm_it3}} of a round~$k$, 
    all the experts $n\in e_k$ 
    have probabilities
    $\pi_n\big(\bm\omega\big[q^c_{k_\infty}\big]\big)$
    and they coincide with each other and do not vary with time, then $q^c_k$ will coincide with them.
\end{enumerate}
\end{Asmp}

We provide arguments for our assumptions to be quite natural and weak.

First, we note that the model does not specify the order in which experts enter the market, but only assumes that trading cannot stop if there is an expert who believes that he or she can make a profitable transaction.

Concerning \ref{srm_it3} and \ref{srm_it6}, the model makes no assumptions on the details of the way the experts trade (in particular, on the matching mechanism), but only the 
assumption that experts whose beliefs are constant and identical to one and the same value, generate the price equal to that value.

Concerning~\ref{srm_it4}, we will see below that the purpose of sharing information is to change others' beliefs and, as a consequence, to profitably change $q^c_{k_\infty}$. In order to achieve this goal in reality, an expert should confirm the truth of the information provided. For example, if the source information is a result of the expert's study, he or she may provide its verified copy: a publication in a peer-reviewed journal. In the above model we identify the expert's information with its verified copy that is transmitted to others. Transmitting unverifiable information is identified with silence. This is a deliberate modeling simplification: here we only show that under certain conditions, fully and accurately disclosing an expert's true information
is at least preferable for him or her to being totally silent. In Appendix~\ref{subsec:sr_reality}, 
we discuss the real situation in its entirety.

Concerning~\ref{srm_it5}, we note that 
the myopia assumption is natural in a context where an expert~$n$ knows nothing about those who are outside $e_k\cup\{n\}$ (even about their existence). 
Certain forms of myopia and risk neutrality are also featured in \cite{FeFoPe2005,PeSa2007}.

\begin{Def}
We say that the market in Assumption~\ref{asmp:SRM} is \emph{efficient} if the following holds for all the rounds $k = 1,\dots,k_\infty$.
\begin{enumerate}[label = \textup{(eff\arabic*)}]
    \item\label{it:eff1} Each expert $n\in e_k$ has shared all his or her information.
    \item\label{it:eff2} All the experts 
    rely in their trading 
    on their mental probabilities of ${\bm\omega\mapsto 1}$. Namely, we have that
    during stage~\ref{srm_it3}, 
        the experts $n\in e_k$ 
        have the probabilities $\pi_n\big(\bm\omega\big[q^c_{k_\infty}\big]\big)$
        and they coincide with the corresponding probabilities 
        ${\bm\pi\bigl(\bm\omega\bigm|\Gamma_{n}^k\bigr)}$.
    
    \item\label{it:eff3} The joining transactions are also regulated by the experts' probabilities 
    of~$\bm\omega\mapsto 1$:
    \begin{itemize}
        \item if there is $n \in \bm E\setminus e_k$ with $\bm\pi\bigl(\bm\omega\bigm|\Gamma_{n}^k\bigr) \ne  q_k^c$, then $k\ne k_\infty$;
        \item an expert $n \in \bm E\setminus e_k$ with $\bm\pi\bigl(\bm\omega\bigm|\Gamma_{n}^k\bigr) =  q_k^c$ does not wish to join $e_k$.
    \end{itemize}
\end{enumerate}
\end{Def}

If the market in question is efficient, then it is in a certain form of a Nash equilibrium. Namely, we have the following result.

\begin{Th}\label{lem:eqlb}
Suppose during all the rounds $k = 1,\dots,k_\infty$, an expert $m \in \bm E$ 
believes 
\begin{enumerate}[label = \textup{(m\arabic*)}]
    \item\label{it:m1} that all the other experts $n\in\bm E\setminus\{m\}$ satisfy~\textup{\ref{it:eff1}} 
    and~\textup{\ref{it:eff2}}; 
    \item\label{it:m3} that he or she cannot resist the concerted actions of others in the sense that the set $e_k\setminus\{m\}$ is not empty and \textup{\ref{srm_it6}} is true even when substituting 
    $e_k\setminus\{m\}$ for $e_k$.
\end{enumerate}

Then the expert~$m$ satisfies~\textup{\ref{it:eff1}}, \textup{\ref{it:eff2}}, and~\textup{\ref{it:eff3}} during all the rounds.
\end{Th}
\begin{proof}
Suppose the expert~$m$ silently joins $e_k$.  Due to~\ref{it:m1}, 
he or she believes 
that during stage~\ref{srm_it3} of the round $k+1$, the probabilities
$\pi_n\big(\bm\omega\big[q^c_{k_\infty}\big]\big)$ 
coincide with $\bm\pi(\bm\omega\mid\Gamma_n^{k+1})$ for $n\in e^{k+1}\setminus \{m\}$, where
$$\Gamma_n^{k+1} = \Gamma_n^k = \Gamma^k.$$
Due to~\ref{it:m3} and the myopia described in the first item of~\ref{srm_it5}, 
the expert~$m$ believes 
$$q^c_{k_\infty} = q^c_{k+1} = \bm\pi\bigl(\bm\omega\bigm|\Gamma_n^{k+1}\bigr) = \bm\pi(\bm\omega\mid\Gamma^{k}) = q_k^c.$$
Therefore, 
we obtain that 
when the expert~$m$ enters the market, he or she has
\begin{equation}\label{eq:sl_cons}
\pi_m\big(\bm\omega\big[q^c_{k_\infty}\big]\big) 
=q^c_{k},
\end{equation}
and thus the expected profit on his or her joining transaction is zero. 
This contradicts the second item of~\ref{srm_it5}, 
and thus $m$ satisfies~\ref{it:eff1}. 

Suppose the expert~$m$ belongs to $e_k$ and is at stage~\ref{srm_it3} of the round~$k$. 
We have just proved that he or she has shared his or her information and believes
$$\Gamma_m^k = \Gamma_n^k = \Gamma^k,\quad n\in e_k.$$ 
The myopia of $m$, together with assumptions~\ref{it:m1} and~\ref{it:m3}, implies that $m$ believes
$$
    q^c_{k_\infty} = q^c_{k} = \bm\pi(\bm\omega\mid\Gamma^{k}) = \bm\pi\bigl(\bm\omega\bigm|\Gamma_m^{k}\bigr).
$$
Namely, he or she has
$$
\pi_m\big(\bm\omega\big[q^c_{k_\infty}\big]\big) = 
\bm\pi\bigl(\bm\omega\bigm|\Gamma_m^{k}\bigr)
$$
and satisfies~\ref{it:eff2}.

If the expert~$m$ satisfies the condition of the first item of \ref{it:eff3} and joins to $e_k$, then he or she will have $\Gamma_m^{k+1} = \Gamma_m^{k}$. 
As we have just proved, the expert~$m$ will have 
$$
    \pi_m\big(\bm\omega\big[q^c_{k_\infty}\big]\big) = 
    \bm\pi\bigl(\bm\omega\bigm|\Gamma_m^{k+1}\bigr) =\bm\pi\bigl(\bm\omega\bigm|\Gamma_m^{k}\bigr) \ne q_k^c.
$$
Then by the third item of~\ref{srm_it5}, 
the round~$k$ cannot be the last.

If the expert~$m$ satisfied the condition of the second item of \ref{it:eff3} and joined to~$e_k$, then by the same reason as above, he or she would have
$
\pi_m\big(\bm\omega\big[q^c_{k_\infty}\big]\big) = q_k^c
$
and we would have a contradiction with the second item of~\ref{srm_it5}. 
\end{proof}

First, we note that item~\ref{it:m3} concerns a robustness property that is fulfilled if, for example, 
the price can be calculated as the median of beliefs. This can be seen as a model of the real situation discussed in Appendix~\ref{subsec:sr_reality}, 
where we
rely (instead of the price) on the so-called market-driven median of experts' beliefs.


Second, we note that it may be expected that the described equilibrium will naturally arise if we publicly offer the experts to trade as though the market will be resolved in accordance with~$\bm\omega$,
and to share their information in order for the system to be able to model this.

Finally, we obviously have the following analogue of~\eqref{eq:final_state}.
\begin{Fact}\label{rem:final_state_sr}
If the market in Assumption~\textup{\ref{asmp:SRM}} is efficient, then 
\begin{equation}\label{eq:final_state_sr}
    \bm\pi\bigl(\bm\omega\bigm|\Gamma_n^{k_\infty}\bigr)
    = q^c_{k_\infty} = \bm\pi(\bm\omega\mid\Gamma^{k_\infty}), \quad n\in\bm E.
\end{equation}
\end{Fact}

This means that for each $n\in\bm E$, the information 
    $\Gamma_n$ is either a part of $\Gamma^{k_\infty}$, or the expert~$n$ does not understand 
    how $\Gamma_n\setminus\Gamma^{k_\infty}$ may affect $\bm\pi(\bm\omega\mid\Gamma^{k_\infty})$.
    It is also important (!) that the information $\Gamma^{k_\infty}$ has been directly revealed by experts, 
    does not requires computing skills to be elicited, and is fully accessible to any external observer.

    Since the subset $e_1\subseteq \bm E$ of experts with $\Gamma_n =\varnothing$ 
    is not empty, we can make a remark similar to Fact~\ref{rem:d_less_N}.
    \begin{Fact}\label{rem:shorten_fs_sr}
    We can shorten~\eqref{eq:final_state_sr} by dropping the second identity: it is a part of the first.
    \end{Fact}
    
    In the $xyz$-example, 
    the experts have certain information about each other, and there are no ignorant experts with only common signals. These two facts do not allow us to incorporate this example directly into the setting of Assumption~\ref{asmp:SRM}.
    But making appropriate adjustments, we come to the following situation.
    \begin{Ex}\label{ex:xyz_sr}
        Suppose $x$, $y$, and $z$ are mutually exclusive and exhaustive outcomes 
        and this fact is known to all the experts in $\bm E$. Let $\bm\omega$, $Y$, and $Z$ be propositional variables corresponding respectively to the propositions that $x$, $y$, or $z$ has been realized.
        Suppose $Y$ and $Z$ are false. Formally, we have
        \begin{gather*}
            \bm\Gamma =\bigcup_{n\in\bm E}\Gamma_n = \{Y,Z\},\quad \bm\gamma\cii{\bm{E}}(Y)= \bm\gamma\cii{\bm{E}}(Z) =0,\\
            \bm\pi(\bm\omega\mid\varnothing) = \tfrac{1}{3},\quad  \bm\pi\bigl(\bm\omega\bigm|\{Y\}\bigr) = \bm\pi\bigl(\bm\omega\bigm|\{Z\}\bigr) = \tfrac{1}{2},\quad \bm\pi(\bm\omega\mid\bm\Gamma) = 1.
        \end{gather*}
        We also assume that for any $n\in\bm E$, the set $\Gamma_n$ is $\varnothing$ or $\{Y\}$ or $\{Z\}$. If the mechanism described in Assumption~\ref{asmp:SRM} arises (in particular, 
        $e_1 \ne \varnothing$) and it is efficient, then we will have $k_\infty = 3$, $\Gamma^3= \bm\Gamma$,
        and $q_3^c =1$.
    \end{Ex}

    Indeed, during the first round ignorant experts generate the price $q_1^c = \tfrac{1}{3}$, and
    one of the experts with $\Gamma_n =\{Y\}$ or $\Gamma_n = \{Z\}$ has to enter the market and to share his or her information. Therefore, the second round results in $q_2^c = \tfrac{1}{2}$,  
    there exists an expert $n\in \bm E\setminus e_2$ with $\Gamma_n^2 = \bm\Gamma$, and one of such experts has to enter the market and to share the remaining information. Thus, we have $\Gamma^3= \bm\Gamma$ and the third round ends with $q_3^c =1$.
    

\subsection{Implementability}\label{sec:impl}
Here we very briefly describe how the ideas presented in Section~\ref{subsec:sr_theory} can be implemented in reality. 
In particular, we discuss
how to achieve the robustness property~\ref{it:m3} from Theorem~\ref{lem:eqlb}.
The detailed description can be found in Appendix~\ref{sec:model} (see also Appendix~\ref{app:pr_mr}, which introduces some necessary notation).
\begin{itemize}
    \item First, we can describe, relying on the detailed data collected from sports prediction markets, a data-driven model of how real experts behave in \emph{regular} continuous prediction markets (see details in Appendices~\ref{subsec:data}--\ref{subsec:mod_meets_dat}). We can show that an algorithm based on that model can be applied to the data in order to get a value that effectively predicts the consensus probability. Such an estimate can play the role of a price that cannot be affected by trading actions of an individual expert. The assumptions of that continuous model have been themselves derived from the data, but at the same time they make perfect sense in terms of modern economic theory.
    \item We can present (see Appendix~\ref{subsec:sr_reality}) not fully mathematically rigorous but quite detailed arguments showing that using our data-driven model of a regular prediction market and the corresponding algorithm for robust estimation of consensus probability mentioned above, we can transfer to reality the principles contained in our theoretical model of a self-resolving market. Thus, we can make a strong case that it is possible to create the desired system.    
\end{itemize}

Now we provide some additional insight into the mentioned data-driven model of a regular continuous prediction market and the corresponding algorithm for estimating the consensus probability. 
In order to combine information exchange in a prediction market and averaging of experts' probabilities, we introduce (see Appendices~\ref{subsec:pm}--\ref{subsec:voting}) the notion of the market-driven median $\bm\mu(\bm\omega, \bm E)$ of experts' opinions. For each 
expert $n \in \bm E$, it takes into account the entire history $\pi_n^t$ of how his or her mental probability has varied during his or her trading and learning (accumulation of information) in a continuous prediction market. It also assigns certain weights to all values $\pi_n^t$ in accordance with the trading activity of each expert $n$ from $\bm E$ at each time $t$. In the definition of the mar\-ket-driv\-en median $\bm\mu(\bm\omega,\bm E)$ time $t$ runs a large hypothetical interval $[t_0,t_\infty]$ during which trading can potentially continue. 
However, it turns out (see Appendices~\ref{subsec:data}--\ref{subsec:mod_meets_dat}) that it does not require much time to effectively estimate $\bm\mu(\bm\omega,\bm E)$: under certain model assumptions, we can derive from market data an estimate~$\hat{\bm\mu}_{t}$ that stabilizes much faster than the clearing price~$q^c_t$. We can also show that within these assumptions the data 
are consistent with the hypothesis that~$\hat{\bm\mu}_{t}$, $q^c_t$, and all~$\pi_n^t$ converge to one and the same value as $t\to t_\infty$ (cf.~\eqref{eq:final_state} given Fact~\ref{rem:d_less_N}).

The mentioned model assumptions are themselves derived from the data, and the key one concerns the utility function of the experts in question. It turns out to be non-standard 
and only locally satisfying the classical von Neumann--Morgenstern axioms. Let lotteries be finite probability distributions over $\mathbb{R}$: they assign 
probabilities to experts' possessions. In particular, experts' trading actions lead to various lotteries. Consider an expert with a budget~$b$ and 
consider all the lotteries $L$ with $m_L< b < M_L$ together with the inaction 
lottery $L_b$ with $m_{L_b} = M_{L_b} = b$. Here $m_L$ and $M_L$ are the minimum and maximum possible possessions of the expert due to~$L$. Thus, we exclude from consideration all the lotteries where 
win is possible and loss is impossible or vice versa. 
We introduce the lottery dependent utility function
\begin{equation}\label{eq:drp}
	u(x\mid \lambda, b, L) \df 
	\frac{1-\exp\big({-\lambda\frac{x-m_L}{M_L-m_L}}\big)}
	{1-\exp\big({-\lambda\frac{b-m_L}{M_L-m_L}}\big)},
	\quad \lambda \in \mathbb{R},\quad x\in [m_L, M_L].
\end{equation}
We also pass to the limit and set $u(b\mid \lambda, b, L_b) = 1$.
When the expert relies on this function to compare the lotteries under consideration, we say that he or she is subject to
\emph{the double relativity principle}: the scale of the possible results~$x$ is considered relatively to its minimum~$m_L$ and maximum~$M_L$, and the scale of utility is considered 
relatively to the expert's status quo in the denominator of~\eqref{eq:drp}. 
We  
note that such functions do not fit into the existent theories of lottery dependent utilities (such as \cite{BeSa1987, Co1992, Sch2001}).
We detail how 
function~\eqref{eq:drp}
describes  experts' behavior in 
Appendix~\ref{sec:model} and 
continue to discuss it (together with appropriate axioms) in Appendix~\ref{sec:details}.


\printbibliography

\appendix
\markboth{APPENDICES}{APPENDICES}
\section{Extension of Section~\ref{subsec:prob}}\label{app:A}
Here we detail what we mean by experts' probabilities and information. There are many ways to assign personal probabilities to individuals: 
see the survey \cite{Fi1986} or the book \cite{PaIn2009} with references to de Finetti, Ramsey, Savage, et~al. 
Here we present a construction that is different from these theories and is more suitable for our goal of combining experts' information as well as for separating probabilities from our unconventional utilities~\eqref{eq:drp}. 
In other words, we are inspired by the viewpoint of \cite{Wi2009} rather than the one reflected in \cite{Na2001}. However we do not regard our construction as a philosophical system, but only as a preliminary model convenient for further considerations.

First, we give its brief description. We simply postulate that for any expert $n \in \bm E$ there exists a unique value in $[0,1]$ that is considered in our models 
as the expert~$n$'s probability of $\bm\omega \mapsto 1$. 
We do \emph{not} fix what each expert means by his or her probability: 
this is determined, together with the probability itself, by objective signals the expert has received during the lifetime. 
Namely, we introduce a system of axioms that implies Assumption~\ref{asmp:exp_prob}, i.e. that all these probabilities are values of a functional 
over $2^{\bm\Gamma}$, where $\bm\Gamma$ is a finite set of \emph{objective} signals, each of which is known to one or more experts in $\bm E$.
These signals are units of raw information that have been generated directly by external events in the experts' empirical experience. Below we formalize each signal as a propositional variable with a value already
assigned to it. In fact, we can formalize signals as abstract elements of a set: we need propositional logic notation only to show 
(see Section~\ref{subsec:prob}) 
the
relation between our construction and the  
classical Bayesian model  
of knowledge, which appears in various forms, e.g., in \cite{Co1946,Ja2003, Au1976, FeFoPe2005, PeSa2007, ArBaSm2018} 

In addition to the existence of a set~$\bm\Gamma$ and a functional over $2^{\bm\Gamma}$, our axioms imply that if 
a signal is known to an expert in $\bm E$ and lies outside~$\bm\Gamma$, then it is either known to all the experts in $\bm E$ at once, or irrelevant to $\bm\omega$ at all.

\subsection{Objective beliefs}\label{subsec:obj_bel} 
We fix a current time~$t_0$ and attribute to it all the objects introduced below until Section~\ref{subsec:pm}.

\begin{Asmp}[existence axiom]\label{asmp:A0}
    We assume that for each expert $n\in\bm E$, there exists a unique value ${\pi_n(\bm\omega) \in [0,1]}$ that is known to $n$ and is referred to
    as the expert $n$'s \emph{mental probability} of $\bm\omega \mapsto 1$.
    The expert~$n$'s probability of $\bm\omega \mapsto 0$ exists in the same sense and equals 
    $1-\pi_n(\bm\omega)$.
\end{Asmp}

We also refer to the values $\pi_n(\bm\omega)$ as 
\emph{beliefs}. 
We provide 
Assumptions 
\ref{A1} and~\ref{A1_1}, which 
are partially inspired by \cite{Wi2009} and 
incorporates the idea of the
objectivity of beliefs. 
	Informally, we would like to say that 
	there exists a finite collection of \emph{objectively} true signals such that an expert's mental probability of $\bm\omega\mapsto 1$ is determined by those of them that are known to him or her. 
The informal concept of ``an objectively true signal''
suggests some dualism and independence from an expert. We formalize this concept as a propositional variable with a value already assigned to it: this value can be either $0$ or $1$ and cannot be different for different experts.



\begin{Asmp}[signals objectivity axiom]\label{A1}
With each expert $n\in\bm E$, we can associate 
a finite set $\Gamma_n$ of propositional variables with an assignment $\gamma_n\assign\Gamma_n$. 
We assume that $\gamma_n \sim \gamma_m$ for every $n,m\in \bm E$. 
We set ${\bm\Gamma \df \bigcup_{n\in \bm E}\, \Gamma_n}$ and denote by $\bm\gamma\cii{\bm{E}}\assign\bm\Gamma$ a unique assignment such that 
$\bm\gamma\cii{\bm{E}} \sim \gamma_n$ for all $n\in\bm E$. 
For any $\Gamma' \subseteq \bm\Gamma$, we can recombine any expert $n\in\bm E$ with the knowledge ${\bm\gamma\cii{\bm{E}}\vert\cii{\Gamma'}\assign\Gamma'}$ 
and obtain a new mental probability 
$\pi_n(\bm\omega\mid\Gamma')$. 
Namely, Assumption~\textup{\ref{asmp:A0}} 
is assumed to remain true for the modified expert~$n$.
In particular, we have
$\pi_n(\bm\omega) = \pi_n(\bm\omega\mid\Gamma_n)$. 
\end{Asmp}

Assumption~\ref{A1} does not suggest the uniqueness of objects introduced in it. Suppose we have some~$\dtilde{\bm\Gamma}$ constructed in accordance with Assumption~\ref{A1}. 
Then any $\widetilde{\bm\Gamma}\subseteq\dtilde{\bm\Gamma}$ can be endowed with the structure inherited from~$\dtilde{\bm\Gamma}$: we can set 
$$
\widetilde{\Gamma}_n \df \dtilde{\Gamma}_n\cap\widetilde{\bm\Gamma}\quad\mbox{and}\quad
\widetilde{\gamma}_n \df \dtilde{\gamma}_n\vert\cii{\widetilde{\Gamma}_n}
$$
and define an expert $n$ recombined with $\Gamma'\subseteq \widetilde{\bm\Gamma}$ within $\widetilde{\bm\Gamma}$
as $n$ recombined with $$\Gamma''\df\Bigl(\dtilde{\Gamma}_n\setminus\widetilde{\bm\Gamma}\Bigr)\cup\Gamma'$$ within $\dtilde{\bm\Gamma}$. We have
$
        \widetilde{\pi}_n(\bm\omega \mid \Gamma') 
    =  
    \dtilde{\pi}_n(\bm\omega \mid \Gamma'').
$
For any $\widetilde{\bm\Gamma}$ and $\dtilde{\bm\Gamma}$ connected in such a way, we write $\widetilde{\bm\Gamma}\sqsubseteq\dtilde{\bm\Gamma}$.

\begin{Asmp}[beliefs objectivity axiom]\label{A1_1}
The mechanisms $\pi_n(\bm\omega\mid\cdot)$ in Assumption~\textup{\ref{A1}} may differ from each other only because of a finite number of the experts' signals that are not included in ${\bm\gamma\cii{\bm{E}}\assign\bm\Gamma}$. 
Namely, we can choose~$\bm\Gamma$ with the following property. For any $\widetilde{\bm\Gamma}$ satisfying 
Assumption~\textup{\ref{A1}} and such that $\widetilde{\bm\Gamma} \sqsupseteq \bm\Gamma$, the corresponding mechanisms $\widetilde{\pi}_n(\bm\omega\mid\cdot)$ do not depend on~$n$: ${\widetilde{\pi}_n(\bm\omega\mid\cdot) = \widetilde{\bm\pi}(\bm\omega\mid\cdot)}$. In particular, we have
a single mechanism ${\bm\pi}(\bm\omega\mid\cdot)$ over $\bm\Gamma$ itself.
\end{Asmp}

The variables in $\bm\Gamma$ with assigned values represent very basic objective signals without any derivations. In other words, the knowledge $\bm\gamma\cii{\bm{E}}\assign\bm\Gamma$ is ontological, not epistemological. We have a finite number of signals that binary encode
the experts' empirical experience consisting of outer events. These ``raw'' signals determine the experts' knowledge on $\bm\omega$ as well as how they derive their probabilities from this knowledge.  
In particular, we may assume that the presence or absence of certain ontological signals in an expert's experience determine his or her epistemological language as well as his or her 
type of rationality or irrationality 
(e.g., Bayesian rationality as in \cite{Co1946,Ja2003}, or rationality based on 
the maximum entropy principle as in \cite{Wi2009,Wi2011}, 
or irrationality as in \cite{TvKa1974,KaTv1977}, 
and so forth).
The ontological nature of the signals justifies the assumption that they may be considered separately from the experts. 
Assumption~\ref{A1_1} means that by ``scooping out'' such signals, we can come to a situation where all the experts are coherent and any signal outside~$\bm\Gamma$ is common to them or is irrelevant to $\bm\omega$. 
This statement acquires a precise meaning due to the following proposition.

\begin{Le}\label{lem:irr} If $\dtilde{\bm\Gamma}$ satisfies Assumption~\textup{\ref{A1}} and $\bm\Gamma_0 \sqsubseteq \dtilde{\bm\Gamma}$, then 
the following
statements are equivalent.
\begin{enumerate}[label = \textup{(st\arabic*)}]
    \item \label{st1} 
    For any $\widetilde{\bm\Gamma}$ such that 
    $\bm\Gamma_0\sqsubseteq\widetilde{\bm\Gamma}\sqsubseteq\dtilde{\bm\Gamma}$, the mechanisms
    $\widetilde{\pi}_n(\bm\omega\mid\cdot)$ do not depend on~$n$.
    \item 
    The mechanisms $\dtilde{\pi}_n(\bm\omega\mid\cdot) = \dtilde{\bm\pi}(\bm\omega\mid\cdot)$ do not depend on~$n$, and for any $R\in \dtilde{\bm\Gamma}\setminus \bm\Gamma_0$, we have
        \begin{align}\nonumber
        &\mbox{either}\quad R \in \bigcap_{n\in\bm E}\dtilde{\Gamma}_n &&&\\\label{eq:irr}
        &\mbox{or}\quad \dtilde{\bm\pi}(\bm\omega\mid\Gamma'\cup\{R\})=\dtilde{\bm\pi}(\bm\omega\mid\Gamma') 
        \quad\mbox{for any}\quad \Gamma' \subseteq \dtilde{\bm\Gamma}. &&&
        \end{align}
\end{enumerate}
\end{Le}
\begin{proof}
    Suppose~\ref{st1} is fulfilled and there exists $R\in \dtilde{\bm\Gamma}\setminus \bm\Gamma_0$ such that $R\notin\dtilde{\Gamma}_m$ for some $m\in\bm E$ and 
    \eqref{eq:irr} is not fulfilled for some $\Gamma' \subseteq \dtilde{\bm\Gamma}$.
     We also know that there exists an expert $n \in \bm E$ such that $R\in\dtilde{\Gamma}_n$. If $\widetilde{\bm\Gamma} = \dtilde{\bm\Gamma}\setminus\{R\}$, then 
     $\Gamma' \subseteq \widetilde{\bm\Gamma}$ and $\widetilde{\pi}_m(\bm\omega\mid\Gamma') \ne \widetilde{\pi}_n(\bm\omega\mid\Gamma')$. 
     We come to a contradiction. The proof of the reverse implication is even simpler, and we leave it to the reader.
\end{proof}

Since we assume that $\bm\Gamma$ satisfies Assumptions~\ref{A1} and~\ref{A1_1}, Proposition~\ref{lem:irr} implies that the value of  $\widetilde{\bm\pi}\bigl(\bm\omega\bigm|\widetilde{\bm\Gamma}\bigr)$ is one and the same 
for all $\widetilde{\bm\Gamma}\sqsupseteq \bm\Gamma$ and we arrive at the following definition.

\begin{Def}\label{def:coll_prob}
We call
$
	\bm p(\bm\omega, \bm E) \df \bm\pi(\bm\omega\mid\bm\Gamma)
$
\emph{the collective probability}.
\end{Def}


\section{Extension of Section~\ref{subsec:pr_mr}}\label{app:pr_mr}
\subsection{Bayesian discrete model continued}\label{subsec:pen_res} 
As we have indicated, the first result of Feigenbaum et al.~\cite{FeFoPe2005, PeSa2007} presented in Proposition~\ref{prop:final_state} 
is not their main contribution: it is a direct consequence of the classical study~\cite{NiBrGe1990}. Here we provide the second (and main) result of~\cite{FeFoPe2005}. 
First, following~\cite{FeFoPe2005} we assume that the experts believe they can completely restore~$\omega$ from~${s\in S}$.
\begin{Asmp}[Feigenbaum et~al's model continued]\label{asmp:FFPS_add}
    We have $$P(\omega = g(s))=1,$$ where $g \colon S \to \{0,1\}$ is a given Boolean function. 
\end{Asmp}
    
The main result of \cite{FeFoPe2005} is the following fact that complements Proposition~\ref{prop:final_state} and generalizes Example~\ref{ex:xyz}. 
\begin{Th}\label{lem:FFPS} 
Under additional Assumption~\textup{\ref{asmp:FFPS_add}}, a necessary and sufficient condition on $g$ that guarantees the identity
\begin{equation}\label{eq:q_eq_b}
q_{k_\infty}^c = g(r)
\end{equation} 
for any~$P$ and any $r \in I^1$ is the following. There must exist weights $w_0,\dots, w_N \in \mathbb{R}$ such that
\begin{equation}\label{eq:threshold}
g(s) = 1\quad\mbox{iff}\quad w_0 + \sum_{n=1}^N w_n s_n \ge 1.
\end{equation}
\end{Th}

For our part, we note that the existence of functions not satisfying~\eqref{eq:threshold} is not a problem in itself: the measure~$P$ must, in addition, satisfy rather special conditions in order for identity~\eqref{eq:q_eq_b} to be violated. 
For example, if $\bm E$ consists of two experts and $g$ is the XOR function, 
then we have the following proposition (it would also be desirable to obtain its extension to the general case). 
\begin{Le}\label{rem:cond_of_failure}
Let $S=S_1\times S_2$ 
and $g(s) = s_1 \oplus s_2$. 
If we have
\begin{equation}\label{eq:cond_of_failure}
	P(s_n =1) = \tfrac{1}{2},\quad n=1,2,
\end{equation}
then $k_\infty = 1$. 
Otherwise, we have~\eqref{eq:q_eq_b}. 
\end{Le}
\begin{proof}
We denote $P_{ij}\df P(s=(i,j))$. 
Under the condition
\begin{equation}\label{eq:cond2_xor_proof}
P_{ij}+P_{i'j'}\ne0\quad\mbox{for}\quad(i,j)\ne(i',j'),   
\end{equation}
we have
\begin{align*}
&
\mbox{(\ref*{eq:cond_of_failure}) holds}
\Leftrightarrow
\begin{cases}
P_{10}=P_{01};\\
P_{00}=P_{11};\\
\end{cases}
\!\!\!\!\Leftrightarrow
\begin{cases}
P_{10}P_{00}=P_{01}P_{11};\\
P_{01}P_{00}=P_{10}P_{11};\\
\end{cases}
\\
&\Leftrightarrow
\begin{cases}
\frac{P_{01}}{P_{01}+P_{00}} = \frac{P_{10}}{P_{10}+P_{11}};\\
\frac{P_{10}}{P_{10}+P_{00}} = \frac{P_{01}}{P_{01}+P_{11}};\\
\end{cases}
\!\!\!\!\Leftrightarrow
\begin{cases}
P(s_2=1 \mid s_1 = 0)=P(s_2=0 \mid s_1=1);\\
P(s_1=0 \mid s_2 = 1)=P(s_1=1 \mid s_2=0);\\
\end{cases}
\\[8pt]
&\Leftrightarrow
\;P(g(s)=1 \mid s_n = 0)=P(g(s)=1 \mid s_n=1),\quad n=1,2.
\end{align*}
This implies that if \eqref{eq:cond_of_failure} and \eqref{eq:cond2_xor_proof} hold, then $k_\infty = 1$, and that 
if \eqref{eq:cond_of_failure} does not hold and \eqref{eq:cond2_xor_proof} holds, then~\eqref{eq:q_eq_b} holds.
If \eqref{eq:cond_of_failure} holds and \eqref{eq:cond2_xor_proof} does not, 
then $P_{10}+P_{01}= 0$ or $P_{11}+P_{00}=0$, because the other four relations in~\eqref{eq:cond2_xor_proof} follow 
from~\eqref{eq:cond_of_failure}, and we have $k_\infty = 1$.
Any violation of \eqref{eq:cond2_xor_proof} implies~\eqref{eq:q_eq_b}, and we have that if \eqref{eq:cond_of_failure} and \eqref{eq:cond2_xor_proof} 
does not hold, then \eqref{eq:q_eq_b} holds.
\end{proof}

\subsection{Continuous prediction markets}\label{subsec:pm}
In a real prediction market, trading goes on continuously in real time and experts' behavior may be complex.

\begin{wraptable}{r}[0pt]{0.35\linewidth}
\centering
\caption{Market state.}\label{PSState}
\begin{tabular}{@{}ccc@{}}
\toprule[1.0pt]
Demand           &          & Supply           \\ \midrule
$S_t^+(q_1)$     & $\ge$    & $0$              \\
$\vdots$         & $\vdots$ & $\vdots$         \\
$S_t^+(q_j)$     & $\ge$    & $0$              \\ \midrule
$S_t^+(q_{j+1})$ & $\ge$    & $S_t^-(q_{j+1})$ \\
$\vdots$         & $\vdots$ & $\vdots$         \\ \midrule
$S_t^+(q_t^-)$   & $>$      & $S_t^-(q_t^-)$   \\ \midrule
$\vdots$         & $=$      & $\vdots$         \\ \midrule
$S_t^+(q_t^+)$   & $<$      & $S_t^-(q_t^+)$   \\ \midrule
$\vdots$         & $\le$    & $\vdots$         \\ \midrule
$0$              & $\le$    & $\vdots$         \\
$\vdots$         & $\vdots$ & $\vdots$         \\
$0$              & $\le$    & $\vdots$         \\ \bottomrule[1.0pt]
\end{tabular}
\end{wraptable}
The following method for elicitation of experts' probabilities goes back to de Finetti~\cite{deFi1937trans} and may give some insight into how prediction markets work in reality.
We ask an expert $n\in\bm{E}$: what is the maximum price of one contract on $\bm\omega\mapsto 1$ at which he or she would agree to buy $S$ contracts? It is easy to see that the expert's answer tends 
to $\pi_n(\bm\omega)$ as $S\to 0$, provided his or her preferences are determined by $\pi_n(\bm\omega)$ and a von Neumann--Morgenstern utility function (see, e.g., the discussion in~\cite{KaWi1988}). On the other hand, the impact of utilities~\eqref{eq:drp} 
does not depend on~$S$ and cannot be eliminated from the answer to de Finetti's question in such a way. 
Nevertheless, in Section~\ref{sec:model} 
we will see how this impact vanishes during trading in a prediction market.

Now we give the following formal definition, which describes a prediction market in reality.
\begin{Def}\label{def:CDA}
\emph{A prediction market} for~$\bm\omega$ 
is a sequence 
of \emph{orders} $o_t = (s, q, n)^{\pm}$ where $t$ 
runs over a finite subset of $[t_0,+\infty)$. Each order means that at time~$t$, an expert $n\in\bm E$ claims that he or she is 
ready to buy (for orders with~$+$) or to sell (for orders with~$-$) $s > 0$ contracts on $\bm\omega\mapsto 1$ 
by a price $q \in (0,1)$.
\end{Def}

We assume that orders are matched with each other in real time by natural rules, i.e. that our prediction market is arranged as 
\emph{a continuous double auction} (CDA). We skip the details of these rules, but note that 
at a moment $t>t_0$ we will see 
a situation shown in Table~\ref{PSState}. Here 
$$0< q_1 < \dots < q_j < q_{j+1}<\dots<q_t^-<\dots<q_t^+ <\dots<1$$ 
are all the prices with non-zero total demand $S_t^+(q)$ or total supply $S_t^-(q)$ of contracts at time~$t$. For any of these prices, the number of concluded transactions and the number of contracts available for immediate selling or buying, are
$$
\min\big(S_t^+(q), S_t^-(q)\big) 
\quad
\mbox{and}
\quad
\big|S_t^+(q)-S_t^-(q)\big|,
$$
respectively.
An expert may withdraw any portion of his or her offer that has not been matched yet.
The bid price~$q_t^-$ is the best price for immediate selling at time~$t$: it is the highest price where demand is greater than supply. 
The ask price~$q_t^+$ is the best price for immediate buying at time~$t$: it is the lowest price where supply is greater than demand.
By $q^c_t$ we denote the current clearing price: the price of the last concluded transaction.

\subsection{Beliefs in a continuous prediction market}
As a basis for further considerations in Section~\ref{sec:model}, we consider a purely hypothetical, but imaginable experiment where all the experts in~$\bm E$ participate in a time-\emph{un}limited prediction market and where they are isolated from external signals 
(i.e., they receive information only from each other, e.g., via the market). 
The isolation of $\bm E$ can be formalized as the constancy of the collective probability in the course of the experiment. 
\begin{Asmp}[continuous prediction market pre-model]\label{asmp:A2}
There is the possibility of an experiment satisfying the following conditions.
\begin{enumerate}
\item\label{A2}
At the moment~$t_0$, all the experts in~$\bm E$ are given the opportunity 
to participate in a CDA-type prediction market for~$\bm\omega$, and they know about that. 
\item\label{A2_1} 
At each moment $t\ge t_0$ of the experiment, we get into Assumptions~\textup{\ref{A1}} and~\textup{\ref{A1_1}}: 
there are raw information $\bm\gamma\cii{\bm{E}}^t\assign\bm\Gamma^t$, the experts' probabilities $\pi_n^t(\bm\omega)$, 
and a single mechanism $\bm\pi^t(\bm\omega\mid\cdot)$
over $\bm\Gamma^t$ 
that generates them. 
We assume that if $t_1<t_2$, then $\Gamma_n^{t_1}\subseteq\Gamma_n^{t_2}$
and $\bm\pi^{t_1}(\bm\omega\mid\Gamma') = \bm\pi^{t_2}(\bm\omega\mid\Gamma')$ for any $n\in\bm E$ and any $\Gamma'\subseteq\bm\Gamma^{t_1}$.
\item\label{A2_2} 
The collective probabilities do not vary throughout the experiment: 
$$\bm\pi^{t}(\bm\omega\mid\bm\Gamma^t) = \bm\pi^{t_0}(\bm\omega\mid\bm\Gamma^{t_0}) = \bm p(\bm\omega,\bm E).$$
\item\label{A2_3} 
There is a moment~$t_\infty > t_0$ when the last transaction is concluded: no orders are matched 
with each other during $(t_\infty, t_\infty+\Delta]$, where 
$\Delta$ is a sufficiently large pre-announced timeout. We complete the experiment at $t_\infty+\Delta$ and, after that, distribute awards once 
we know the value of~$\bm\omega$.
\end{enumerate}
\end{Asmp}

We refer to the market from Assumption~\ref{asmp:A2} as \emph{the unlimited market}. 

\begin{Def}\label{qDef} We call the price $\bm q(\bm\omega,\bm E) \df q^{c}_{t_\infty}$
of the last concluded transaction in the unlimited market 
\emph{the ultimate price}.
\end{Def}

\subsection{Trading as weighted voting}\label{subsec:voting} It is worth noting that here we only introduce a certain convenient terminology and do not establish 
any relations with the voting theory, which is a part of the social choice theory. 

Consider the total risks 
$$
V^+_t(q) \df qS^+_t(q)\quad\mbox{and}\quad V^-_t(q) \df (1-q)S^-_t(q) 
$$
for prices $q\in (0,1)$ and times $t>t_0$ (see Table~\ref{PSState}). We call units constituting these quantities 
\emph{votes for} $\bm\omega\mapsto 1$ and \emph{votes for} $\bm\omega\mapsto 0$, respectively. 
Due to discrete nature of any operations with real or play money in any real or imaginable market, we can normalize $V^+_t(q)$ and $V^-_t(q)$ in such a way that 
they
will become integers and each vote will correspond to a single expert who has placed it.
Let $\mathcal{V}_t$ be the set of size $\sum_{q \in (0,1)} \big(V^+_{t}(q)+V^-_{t}(q)\big)$ of all the votes at time~$t$.
We note that since the experts may withdraw non-matched votes, 
it is possible that $\mathcal{V}_{t_1} \nsubseteq \mathcal{V}_{t_2}$ for $t_1 < t_2$.
Suppose we are in the setting of the unlimited market from Assumption~\ref{asmp:A2}. Let $t\in[t_0, t_\infty]$. 
We can assign to each vote $v \in\mathcal{V}_{t}$ the belief $\pi_v = \pi_n^{\tau}(\bm\omega)$ with which 
the corresponding expert $n\in\bm E$ has placed it at time~$\tau\in [t_0,t]$.
\begin{Def}\label{muDef}
We call
$
	\bm\mu(\bm\omega, \bm E) \df \med\{\pi_v\}_{v \in\mathcal{V}_{t_\infty}}
$
\emph{the market-driven median} of beliefs.
\end{Def}

The market-driven median weighs the experts' beliefs and takes into account information exchange between them.
We can demonstrate that 
it is a better choice than the mean 
of~$\pi_v$. Suppose there is an expert who has understood that $\bm\omega$ is true and, therefore, is ready to spend his or her entire budget on contracts on the truth of~$\bm\omega$, buying them at any price $q < 1$. If he or she buys contracts at a price $q>0.5$ from another expert, then such a transaction will generate $V^+$ votes associated with the probability $\pi_1 = 1$ and $V^- < V^+$ votes associated with the probability $\pi_2 < 1$. Therefore, the median of 
the corresponding probabilities $\pi_v$ will be equal to $1$, while the mean will not.

\subsection{Consensus conjecture}
Here we formulate a conjecture that is a continuous counterpart of both
\begin{itemize}
\item
relation~\eqref{eq:final_state}, which can be reduced to its first identity (see Fact~\ref{rem:d_less_N}),  
\item
identity~\eqref{eq:q_eq_b} in Theorem~\ref{lem:FFPS} of Feigenbaum et al. 
\end{itemize}
\begin{Conj}\label{conj:cons}
For any $n\in\bm{E}$, we have
\begin{equation*}
\pi_n^{t_\infty}(\bm\omega) = \bm\mu(\bm\omega,\bm E) = \bm q(\bm\omega,\bm E) = \bm p(\bm\omega,\bm E).
\end{equation*} 
\end{Conj}

We announce that in Section~\ref{sec:model} we build a model that allows us to elicit estimates 
$\hat{\bm\mu}_t$ of $\bm\mu(\bm\omega,\bm E)$ from real market data accumulated during intervals $[t_0,t]$. Inter alia, the model has a parameter~$\sigma_t$
that measures how beliefs~$\pi_v$, $v \in\mathcal{V}_t$, fluctuate on $[t_0, t]$. The estimates~$\hat{\bm\mu}_t$ quickly stabilize (become almost constant), while the estimates $\hat\sigma_t$ decrease. 
We note that the announced model 
implicitly implies that if we consider~$\pi_v$ for $v \in \mathcal{V}_{t_\infty}$ as time series, then $\pi_v$ and $q^c_t$ will stop fluctuating simultaneously at the same level $\bm\mu(\bm\omega,\bm E) = \bm q(\bm\omega,\bm E)$. 
We may assume that after this moment, the experts will have identical beliefs that will remain stable under information exchange. 
As in the discrete setting, the question of the coincidence between this stable shared belief and $\bm p(\bm\omega,\bm E)$ requires a separate discussion (see, e.g., the next section~\ref{subsec:emp_evid}).

It is also worth noting the following.
Let~$f$ stand for various static aggregation schemes that take as arguments 
only the initial probabilities $\big\{\pi_n^{t_0}(\bm\omega)\big\}_{n\in\bm E}$ 
(without involving external weights or information exchange between the experts). 
There is an extensive literature on the choice of~$f$ in various situations. See, e.g., the survey~\cite{ClWi2007} or the recent article~\cite{ArBaSm2018}.
We make general comments on the comparison of $\hat{\bm\mu}_t$, $q_t^c$, and static schemes~$f$ 
as ways to approximate $\bm p(\bm\omega, \bm E)$. The prices~$q_t^c$ oscillate because of the fluctuations of $\pi_v$, while 
the estimates~$\hat{\bm\mu}_t$ remove, like any averaging, these fluctuations. On the other hand, 
unlike static averages~$f$, the estimates~$\hat{\bm\mu}_t$ take into account information exchange between the experts: they may work in situations like the $xyz$-ex\-am\-ple (see Section~\ref{subsec:pr_mr}), 
where any~$f$ returns nothing better than~$1/2$. Thus, the estimates~$\hat{\bm\mu}_t$ combine the advantages of $q_t^c$ and $f$.\footnote{
We also note that $f$ do not provide a way to operationally elicit 
its arguments $\{\pi_n^{t_0}(\bm\omega)\}_{n\in\bm E}$ from the experts, while the model for calculating $\hat{\bm\mu}_t$ solves the problem of operationalization as well.
} See also the studies \cite{PaSo2019,PrSeMc2017,Pe2022}, which propose alternative approaches to aggregating opinions.

\subsection{Empirical evidence for \texorpdfstring{$\bm p = \bm q$}{p=q}}\label{subsec:emp_evid}
There is a crude (but supported by a vast amount of data) argument in favor of the hypothesis
\begin{equation}\label{eq:pq}
\bm p(\bm\omega,\bm E)=\bm q(\bm\omega,\bm E).
\end{equation}
Namely, there are many studies revealing the fact that last observed prices in prediction markets are well-calibrated (but with certain reservations; see, e.g., \cite{PaCl2013}). 
We refer to the articles \cite{PeLaGi2001, SeWo2004, RoNo2006} that confirm this fact both for prediction markets with real money and with play money. However the calibration is only indirectly related to the hypothesis in question: even if we reasonably suppose that $\bm p$ 
and $\bm q$ are both well-calibrated (i.e. that $\bm q$ is not worse than the last observed price and $\bm p$ is not worse than $\bm q$), this will not imply the identity between them. Concerning calibration and refinement (a more subtle measure of forecast accuracy), 
see \cite{DeGrFi1982} or an exposition in \cite[Section~10.4]{PaIn2009}. 

In order to provide a direct empirical verification of the hypothesis, we need to be able to directly calculate~$\bm p(\bm\omega,\bm E)$. This is partially possible only in a controlled experiment where we regulate (to some extent) the initial information of each expert. 
An interesting example is the study \cite{AlKiPf2009} (see also the classic studies \cite{PlSu1988, FoLu1990}), where the authors also support the idea of utilizing play-money prediction markets in the life sciences. Their experiment simulates the search for an answer to a scientific question (in what order three genes activate each other). There are differences from our setting: experts trade contracts of six types that correspond to six mutually exclusive outcomes $\Omega = \{\omega_1,\dots,\omega_6\}$, an automated market maker is used (see~\cite{Ha2007}), and new external information is revealed to experts in the course of trading (that means that the collective probabilities vary). Nevertheless, these differences do not reduce the importance of the discussed experiment to us. 
We briefly describe the situations that have been modeled in the course of it.

\begin{description}
\item[Public setting] 
The initial information on~$\Omega$ is public from the outset.
We note that even in an experiment where the information is under control, it is impossible to make all $\gamma_n \assign \Gamma_n$ identical. 
In addition to direct information on~$\Omega$, each $\gamma_n$ contains personal signals that determine how the corresponding expert processes such information. This explains why in \cite{AlKiPf2009}, the prices have oscillated even in the public setting. Nevertheless, in such a setting the prices have become close to the collective probabilities that means that prediction markets at least do not ``lose'' public information. 
\item[Private setting]
The information of each expert is private and is transmitted to others through the prices only.
Six experimental markets have been organized in the private setting. In half of the cases, the markets have succeeded to aggregate information: the results do not contradict hypothesis~\eqref{eq:pq}. In three remaining markets the prices have not converged to the collective probabilities: better approximations have been obtained by deriving theoretical beliefs from the initial private information of each expert and by averaging them. It is reasonable to assume that estimates of the corresponding market-driven medians could give better results, provided we could elicit them from the market data.
\item[Mixed setting]
The information of each expert is initially private, but becomes public after a while (direct information exchange is simulated).
All six corresponding markets have succeeded to aggregate information in such a setting. The results are almost the same as in the public setting. 
\end{description}

Thus, we have empirical evidence that prediction market prices are able to aggregate information, but may fail to do so if they are the only information channel between the experts. As might be expected, direct information exchange can greatly improve the reliability of prediction markets.

\section{Extension of Section~\ref{sec:impl}}\label{sec:model}
Here we present the data-driven continuous model of regular prediction markets that has been repeatedly mentioned above, and show that within its framework, our data~\cite{Data2019} do not contradict Conjecture~\ref{conj:cons}. We also show that there exist estimates $\hat{\bm\mu}_t$ of $\bm\mu(\bm\omega,\bm E)$ that converge rapidly and remain stable under trading actions of experts. 
We discuss how these facts can be exploited in order to create a real continuous self-resolving prediction mechanism with direct information exchange.

%
%
\subsection{Data}\label{subsec:data}

Our goal is the creation of a collective intelligence based on play-money prediction markets and intended for 
use in science, medicine, and engineering. 
Nevertheless, for the construction and verification of our continuous model we need the extensive and detailed data~\cite{Data2019} collected from CDA-type real-money prediction markets related to sporting events. The resulting model seems fairly general and independent of whether 
$(\bm\omega,\bm E)$ relates to sports or to science. Eventually we come to the consideration of a continuous self-resolving prediction market with play money.

\begin{wraptable}{r}[0pt]{0.45\linewidth}

\centering
\caption{Data example.}
\label{tab:data_example}
\begin{tabular}{@{}cccc@{}}
\toprule[1.0pt]
q                & $S^+$  & \multicolumn{1}{l}{}   & $S^-$  \\ \midrule
0.4167           & 6.6    & \multirow{2}{*}{$\ge$} & 0.0    \\
0.4505           & 501.9  &                        & 0.0    \\ \cmidrule(l){2-4} 
0.4808           & 315.5  & \multirow{2}{*}{$\ge$} & 20.0   \\
0.4854           & 544.3  &                        & 149.4  \\ \cmidrule(l){2-4} 
$q^-_t = 0.4902$ & 261.5  & $>$                    & 261.1  \\ \cmidrule(l){2-4} 
0.4950           & 588.0  & $=$                    & 588.0  \\ \cmidrule(l){2-4} 
$q^+_t = 0.5000$ & 304.1  & $<$                    & 668.3  \\ \cmidrule(l){2-4} 
0.5025           & 450.3  & \multirow{7}{*}{$\le$} & 1007.9 \\
0.5051           & 8.2    &                        & 8.2    \\
0.5076           & 926.3  &                        & 926.3  \\
0.5102           & 1091.9 &                        & 1561.0 \\
0.5128           & 405.8  &                        & 630.9  \\
0.5155           & 73.6   &                        & 73.6   \\
0.5181           & 48.4   &                        & 63.8   \\ \cmidrule(l){2-4} 
0.5236           & 0.0    & \multirow{2}{*}{$\le$} & 13.0   \\
0.5263           & 0.0    &                        & 2642.5 \\ \bottomrule[1.0pt]
\end{tabular}
\end{wraptable}

The data concern 24 soccer games of 2013--2014. Several markets are considered for each game, and each market contains several mutually exclusive outcomes. 
If there are more than $2$ outcomes in a market, then compound matching is possible and we cannot, generally speaking, reduce the situation to separate markets associated with separate propositional variables. 
But we have collected data for each outcome of each market separately.
Such data contain the dynamics of demand and supply of contracts on the occurrence of the outcome (together with the volumes of concluded transactions).
There may be transactions that have been concluded due to compound matching, but this is not reflected 
in the data and is not important for our model. Thus, the data for each outcome of each market fit into 
the pattern from Table~\ref{PSState}, and it is possible to treat them as data of a separate market for a separate propositional variable (we can merge data if there are exactly $2$ complementary outcomes in a market).  

Let $\bm\omega$ represent one of $576$ outcomes from one of 96 markets related to one of 24 soccer games mentioned above. Let $[t_0,T]$ 
be the interval\footnote{At the pre-announced time~$T$ all non-matched votes have been returned to the owners, and very shortly after that the corresponding soccer game have started (together with a so-called live market).} where the corresponding market has existed. 
Then the data for $\bm\omega$ can be represented as a collection of time series $V^{\pm}_t(q)$ of volumes in dollars that 
have the same meaning as in Section~\ref{subsec:voting}. 
Here $q$ runs over the finite (but quite dense)
grid $Q \subset (0,1)$ of all prices available for trading, and $t$ runs over an observation interval 
$[t_1,t_2] \subset [t_0,T]$ (where the changes of market state have been checked with frequency of approximately $1$~second and $T - t_2 < 20\mbox{ sec.}$). 
We note that for each $t\in[t_1,t_2]$, the corresponding aggregated data represent the whole interval $[t_0,t]$, 
not only $[t_1,t]$: we have started by obtaining $V^{\pm}_{t_1}(q)$ and, after that, have scanned subsequent changes.

Repository~\cite{Data2019} contains our PostgreSQL data and R scripts as well as the R workspace \path{CI.RData} where 
all the scripts and some data for two soccer events (see the data frames \verb+data12+ and \verb+data13+) have been preloaded (see also~\cite{RCT2021, Wi2016}). Other data can be loaded using the functions \verb+GetData+ or \verb+GetDataFromCSV+. 
All calls given in this article should be made from \path{CI.RData}.
For example, the following call returns the data in the $S^\pm$-form for an outcome of a French soccer event.
\begin{verbatim}
    AggData(dat = data13,   # 08 March 2014, Guingamp vs. Evian TG
            outcome = 1,    # Guingamp win outcome
            time = 30*60)   # t = t_1 + 30 min. 
\end{verbatim}
The result is partially shown in Table~\ref{tab:data_example}. It is fully consistent with Table~\ref{PSState}.

\subsection{Continuous model assumptions: probabilities}\label{subsec:mod_prob} 
Suppose there is a prediction market for~$\bm\omega$ that 
exists during a pre-announced interval $[t_0,T]$, generates volumes~$V^{\pm}_{t}(q)$, and yields a sample $\bm e \subseteq \bm E$ of all the experts who add some orders to it during $[t_0,T]$. 
The first two model assumptions concern the experts' probabilities and connect the model with the pre-model described in Assumption~\ref{asmp:A2}.
\begin{Asmp}[connection with the pre-model]
\label{M_0}
For the market just introduced, conditions \ref{A2}--\ref{A2_2} from Assumption~\ref{asmp:A2} hold for $t\in [t_0,T]$ and for $\bm e$ substituted for~$\bm E$.
\end{Asmp}
We refer to the market that is introduced above and satisfies Assumption~\ref{M_0}, as \emph{the regular market}.

In the regular market, we can, as in Section~\ref{subsec:voting}, consider $\pi_v$ where $v$ runs over the sets~$\mathcal{V}_t$ of normalized votes.
The next assumption postulates (inter alia) that in spite of a limited time and number of participants, the regular market 
is representative for the unlimited market 
from Assumption~\ref{asmp:A2} with respect to $\bm\mu(\bm\omega,\bm E)$.
Let $f_{\mu,\sigma}$ and $F_{\mu,\sigma}$ be the density and cumulative distribution functions of $\mathcal{N}(\mu,\sigma^2)$, 
and let~$\delta_x$ be the delta measure concentrated at~$x$.
\begin{Asmp}[beliefs normality and representativeness]\label{M_1}
Let $t \in [t_0,T]$. 
We assume that votes $v \in \mathcal{V}_t$ have been generated by separate orders and that
the values $\{\pi_v\}_{v\in \mathcal{V}_t}$ 
are i.i.d. 
with distribution $\mathcal{N}_{[0,1]}(\mu,\sigma_t^2)$. It is a modifications of the normal distribution with all the values outside the unit interval reset to the points~$0$ and~$1$, i.e. with the density 
functions
$$
	\delta_0(x)\,F_{\mu,\sigma_t}(0) + \mathds{1}_{(0,1)}(x)\, f_{\mu,\sigma_t}(x) + \delta_1(x)\,(1-F_{\mu,\sigma_t}(1)).
$$
We assume that $\mu$ does not depend on $t$ and represents 
the market-driven median that would generated by the unlimited market 
if it was started instead of the regular market 
at the same point~$t_0$:
$$
\bm\mu(\bm\omega,\bm E) = 
\begin{cases}
0,& \mu \le 0;\\
\mu,& 0\le \mu\le 1;\\
1,& \mu \ge 1.
\end{cases}
$$
\end{Asmp}

The sample~$\bm e$ is what we initially have or what we can control. 
What a population~$\bm E$ it represents is a matter of common sense and situation, not mathematics. At least, we must have reason to believe that $\bm p(\bm\omega,\bm E) = \bm p(\bm\omega,\bm e)$. In particular, we may set $\bm E \df \bm e$. 
If we talk about our data~\cite{Data2019}, they have been collected from the largest platform with sports CDA markets. Its numerous users represent most of the countries, and we may
assume that $\bm E$ is the population of all the people except some insiders.\footnote{We assume that coaches or team members do not participate in the trade.}


Next, we note that in fact, Assumption~\ref{M_1} generates a collection of models: for each interval $[t_0, t]$, we have its own averaged dispersion~$\sigma_t$ of beliefs. 
Each such model concerns the aggregated data~$V^{\pm}_{t}(q)$ that include \emph{no} time-related information on how these volumes have been produced during $[t_0,t]$.
In fact, we do not introduce here a stochastic process that accurately models the dynamics of~$\pi_v$, 
but rather make a certain ergodicity assumption: loosely speaking, there exist $\sigma_t$ such that the quantities
$$
    \biggl|\frac{1}{|\mathcal{V}_t|}\sum_{v\in\mathcal{V}_t} \mathds{1}_{[0,\theta]}(\pi_v) - F_{\mu,\sigma_t}(\theta)\biggr|,\quad \theta\in [0,1],
$$
become small as $t$ increases 
(cf. \cite[Section~V.3.4, Problem~3]{Sh1989trans}).

Concerning the distributions $\mathcal{N}_{[0,1]}(\mu,\sigma_t^2)$, we note that they have the desired property that 
the values $0$ and $1$ are not excluded from consideration. 
See also \cite[Figure~5]{WoZi2006}, which empirically justifies our choice of distribution.

\subsection{Continuous model assumptions: utilities}\label{subsec:mod_util}
Here we rely on the lottery-dependent utility function $u(x\mid \lambda, b, L)$ that incorporates the double relativity principle and is defined by formula~\eqref{eq:drp}. 
We continue to discuss such functions in Section~\ref{subsec:util}.

\begin{Asmp}[experts' utilities]\label{M_2}
Let $o_\tau = (s,q,n)^{\pm}$, $\tau\in [t_0,T]$, be any order in the regular market. 
By $b_n^\tau$ we denote the budget of the expert~$n$ that he or she has had immediately \emph{before} placement of~$o_\tau$. 
We assume that in addition to $\pi_n^\tau(\bm\omega)$, the expert~$n$ associates with $o_\tau$ some values 
$$\rho(o_\tau) \in (0,1]\quad\mbox{and}\quad 1-\rho(o_\tau) \in [0,1)$$ 
that play, in accordance with inequality~\eqref{eq:M_2} below, the role of probabilities of
the events ``$o_\tau$~will be completely matched'' and 
``$o_\tau$ will remain completely unmatched'', respectively. 
The alternative ``$o_\tau$~will be partially matched'' is not considered at all. We consider the lottery~$L$ that corresponds to the possible variations of the possession~$b_n^\tau$ due to the order $o_\tau$. 
We assume that placing $o_\tau$, the expert~$n$ 
believes (in terms of the double relativity) that $L$ is better than inaction. Namely, if $o_\tau$ is a buy order, then 
either $\pi_n^\tau(\bm\omega) = 1$ or
\begin{equation}\label{eq:M_2}
\rho(o_\tau)\, U(o_\tau)
+(1-\rho(o_\tau))\,u\big(b_n^{\tau} \,\big |\, \lambda_n^\tau, b_n^{\tau}, L\big)> 1,
\end{equation}
where $U(o_\tau)$ is the expected utility of the order~$o_\tau$ as though it has been completely matched: 
\begin{equation}\label{eq:U}
U(o_\tau) = \pi_n^\tau(\bm\omega)\, u\big(M_L \,\big |\, \lambda_n^\tau, b_n^{\tau}, L\big) 
+(1-\pi_n^\tau(\bm\omega))\, u\big(m_L \,\big |\, \lambda_n^\tau, b_n^{\tau}, L\big),
\end{equation}
where 
$$\pi_n^\tau(\bm\omega)\in(0,1),\quad M_L = b_n^\tau + (1-q)s,\quad m_L = b_n^\tau - qs,\quad\mbox{and}\quad \lambda_n^\tau \in\mathbb{R}.$$
For a sell order, we 
swap $q$ with $1-q$ and $\pi_n^\tau(\bm\omega)$ 
with $1-\pi_n^\tau(\bm\omega)$.
\end{Asmp}
We make some comments on this assumption. First, the quantity~$b_n^\tau$ is the result that the expert~$n$ has expected, before placement of~$o_\tau$, from the market in the worst-case scenario. In other words, $b_n^\tau$ is the amount available for trading at time~$\tau$. 
The expert~$n$ may be already involved in a lottery at time $\tau$ due to previous orders, but 
we assume that he or she thinks locally and decides whether to leave the amount~$b_n^\tau$ untouched or 
to use a part of it in a new lottery (which is considered separately). 
We will see below that nothing depends on the values~$b_n^\tau$ at all. 

Second, the assumption that the expert~$n$ does not believe the order~$o_\tau$ may be partially matched is consistent with the fact that in Assumption~\ref{M_1} we represent the market as a flaw of small ``unit'' orders.

By combining formula~\eqref{eq:drp} for $u(x\mid \lambda, b, L)$ 
with 
relations~\eqref{eq:M_2} and~\eqref{eq:U} and by making simple equivalent transformations, we reach the following consequence.
\begin{Le}\label{lem:cenz_mech}
Let $t\in[t_0,T]$, and let $v\in\mathcal{V}_t$ be one of the votes constituting $V^{+}_t(q)$ for some $q\in(0,1)$.
Consider the order $o_\tau = (s,q,n)^{+}$ that have generated~$v$ at time $\tau\in[t_0,t]$.
We have
\begin{equation}\label{eq:Upls}
U(o_\tau) = U^+\bigl(\pi_n^{\tau}(\bm\omega), q, \lambda_n^\tau\bigr)\df \frac{\pi_n^{\tau}(\bm\omega)}{\theta^+_{\lambda_n^\tau}(q)},
\quad
\mbox{where}\quad
\theta^+_{\lambda_n^\tau}(q) \df \frac{1-\exp({-\lambda_n^\tau q})}{1-\exp({-\lambda_n^\tau})},
\end{equation}
and $o_\tau$ satisfies Assumption~\textup{\ref{M_2}} if and only if
\begin{equation}\label{eq:theta_plus}
\pi_v =\pi_n^{\tau}(\bm\omega)
>\theta^+_{\lambda_n^\tau}(q). 
\end{equation}

If $v$ is one of $V^{-}_t(q)$ votes and $o_\tau = (s,q,n)^{-}$ is the corresponding order, then 
we have
\begin{equation}\label{eq:Umns}
U(o_\tau) = U^-\bigl(\pi_n^{\tau}(\bm\omega), q, \lambda_n^\tau\bigr)\df \frac{1-\pi_n^{\tau}(\bm\omega)}{1-\theta^-_{\lambda_n^\tau}(q)},
\quad\mbox{where}\quad
\theta^-_{\lambda_n^\tau}(q) \df \frac{1-\exp(\lambda_n^\tau q)}{1-\exp(\lambda_n^\tau)},
\end{equation}
and $o_\tau$ satisfies Assumption~\textup{\ref{M_2}} if and only if
\begin{equation}\label{eq:theta_minus}
\pi_v < \theta^-_{\lambda_n^\tau}(q).
\end{equation}
\end{Le}
Note that 
in Proposition~\ref{lem:cenz_mech}, nothing depends on $s$, $b_n^{\tau}$, or $\rho(o_\tau)$.

Hereinafter, we use some non-standard but convenient notation for certain limit relations. By the symbol $\searrow$ we denote the relation ``decreases to''. Formally speaking, ``$y\searrow b$ as $x\searrow a$'' means that $y$ is an increasing function of $x$ and~$\lim\limits_{x\to a+} y= b$.
\begin{Asmp}[connecting all together]\label{M_3}
Suppose Assumptions 
\ref{M_1} and~\ref{M_2} hold for our regular market. Fix $t\in[t_0,T]$ and suppose $\tau$ runs over $[t_0,t]$. We consider all the orders~$o_\tau$ that 
have generated the votes in~$\mathcal{V}_t$. In Assumption~\ref{M_2}, we take the identical parameters 
$\lambda_n^\tau = \lambda_t < 0$ for all such~$o_\tau$. In addition, we assume that there is a fixed relation between $\lambda_t$ and the parameters $\mu$ and~$\sigma_t$ (see Assumption~\ref{M_1}), provided $\mu \in (0,1)$. Namely, we have $\lambda_t = \lmd(\mu,\sigma_t)$, where 
$$\lmd\colon(0,1)\times(0,+\infty) \to (-\infty,0)$$ is a certain function such that 
\begin{equation}\label{eq:lmd_prop}
|\lmd(\mu,\sigma)|\searrow 0 \quad\mbox{as}\quad \sigma\searrow 0.
\end{equation}
\end{Asmp}

Similar to the situation with~$\sigma_t$, for each time $t$ we have a specific 
parameter~$\lambda_t$ applied to the whole interval~$[t_0,t]$. 
Subsequently, we use a specific function~$\lmd(\cdot,\cdot)$ that is derived from 
certain conditions of economic equilibrium inspired by \cite{Ma2006,WoZi2006}. It is constructed for not too large $\sigma$, but this does not prevent its use in any practical situation. Details are provided in Section~\ref{subsec:econ_equil}. 

We have the following simple fact.
\begin{Fact}\label{rem:gap_dec}
Under Assumption~\textup{\ref{M_3}}, we have $\theta^+_{\lambda_t}(q)<q<\theta^-_{\lambda_t}(q)$ and 
$$|q -\theta^\pm_{\lambda_t}(q)| \searrow 0 \quad\mbox{as}\quad \sigma_t\searrow 0.$$
\end{Fact}
We can give the following interpretation of this fact.
The less the dispersion of beliefs is, 
the less each expert~$n$ expects that he or she does not know something that others do and that his or her probability can be refined. 
High confidence in his or her probability causes 
$n$ to be less deviating from it and to be less inclined to trade at unfavorable prices: to buy at $q > \pi_n^\tau(\bm\omega)$ or to
sell at $q < \pi_n^\tau(\bm\omega)$. An extreme example is where all the experts reliably know that $\bm\omega$ represents a fair coin toss, and have no reason to be risk-seekers.

This can also be treated as follows. A possible gain or loss in a lottery connected with~$\bm\omega$ 
makes the outcomes $\bm\omega\mapsto 1$ and $\bm\omega\mapsto 0$ positive or negative for an expert~$n$ who measures this lottery. 
This fact leads to optimism bias~\cite{We1980}. In other words, the expert~$n$'s mental probabilities of winning and losing the lottery does not coincide with his or her probabilities
$\pi_n^\tau(\bm\omega)$ and $1-\pi_n^\tau(\bm\omega)$ of the underlying outcomes. 
In this interpretation, Fact~\ref{rem:gap_dec} implies that less ambiguity leads to less optimism bias. 
For classical von Ne\-um\-ann--Mor\-gen\-stern utilities, 
a similar relation between optimism/pessimism, risk seeking/aversion, and ambiguity is formalized in~\cite{DiPoRo2017} 
(see also the experimental study~\cite{WeSo2014}).

The optimistic behavior discussed above does not contradict ambiguity aversion~\cite{El1961}, which is another well known behavioral characteristic of
decision makers. 
Namely, 
an expert can always choose to wait until his or her comparative ignorance~\cite{FoTv1995}, causing the feeling of ambiguity,
decreases. This is 
\emph{not} the same as to choose inaction and to leave the market: we do not even measure the waiting alternative in terms of utility. 
The expert's optimism also decreases, but at some point a certain lottery may become preferable to further waiting.
We will see further that~$\sigma_t$ decreases over time together with the overall comparative ignorance. 
The tendency for experts to wait is confirmed by the fact that the intensity of trading increases over time.


\subsection{Continuous model meets real data}\label{subsec:mod_meets_dat}
\begin{figure}
\includegraphics[scale = 0.6]{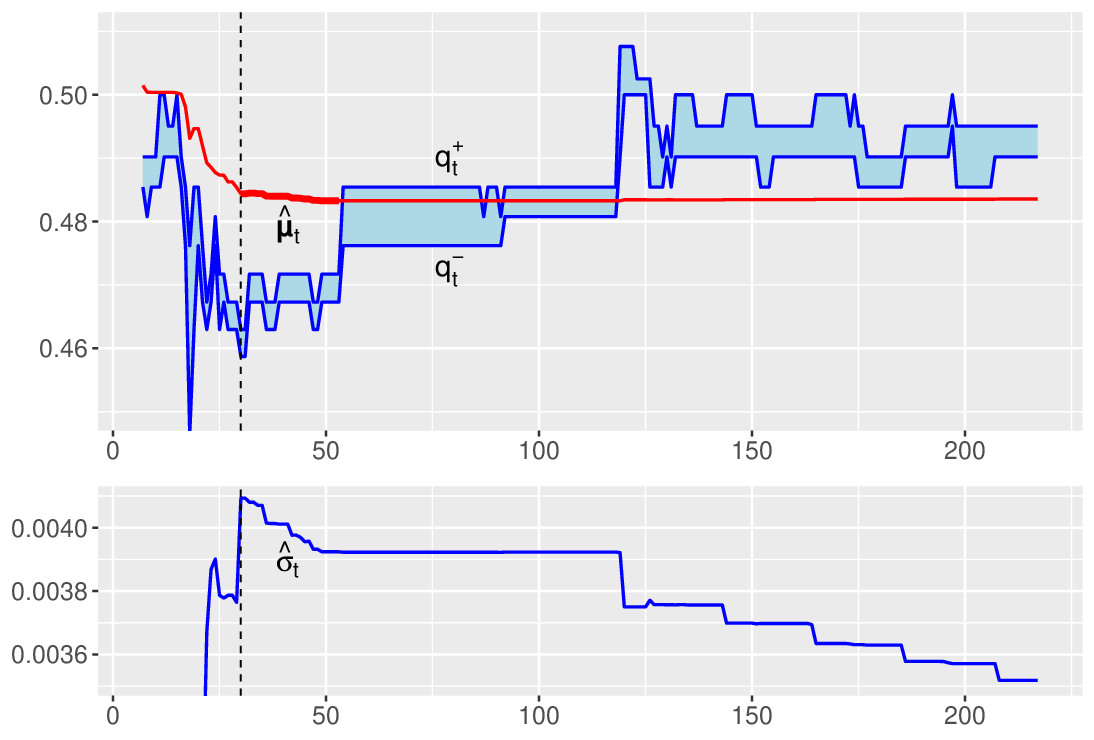}
\caption{Estimates in operational time.}
\label{fig:mu_sgm_gr}
\end{figure}
Suppose $\bm\omega$ represents one of the outcomes described in Section~\ref{subsec:data} and  
the volumes~$V_t^\pm(q)$ represent the corresponding real data in~\cite{Data2019}.
Let $t$ be from the observation interval $[t_1,t_2]$. 
We can calculate 
the probability of the event that $\pi_v$ satisfy inequalities~\eqref{eq:theta_plus} 
and~\eqref{eq:theta_minus} where $v$ runs over~$\mathcal{V}_t$ and all $\lambda_n^\tau$ coincide with single~$\lambda$. We interpret the logarithm of
this probability as a log-likelihood function:
\begin{multline}\label{eq:logL}
    \log L(\mu,\sigma\mid\lambda)\\ \df 
    \sum_{q\in Q} \Big(V_t^+(q)\log\big(1-F_{\mu,\sigma}\big(\theta^+_{\lambda}(q)\big)\big) +
    V_t^-(q)\log F_{\mu,\sigma}\big(\theta^-_{\lambda}(q)\big)\Big).
\end{multline}
Certain issues concerning the validity of applying the function $\log L$ are discussed in Section~\ref{subsec:market_dyn}.
It turns out that we can choose $\hat\lambda_t < 0$ for which 
a numerical optimization of $\log L\bigl(\mu,\sigma\bigm|\hat\lambda_t\bigr)$ over $\mu$ and $\sigma$ gives estimates $\hat{\bm\mu}_t$ and $\hat\sigma_t$ such that
${\hat\lambda_t = \lmd(\hat{\bm\mu}_t,\hat\sigma_t)}$.
Details of the corresponding algorithm can be found in Section~\ref{subsec:econ_equil}.
Now we employ this approach to \verb+data13+ from \path{CI.RData} (see~\cite{Data2019} and Section~\ref{subsec:data}) and show how the estimates in question vary over time. 
We make the following calls.
\begin{verbatim}
    gr <- GetGraphOperTime(dat = data13, outcome = 1, voldlt = 1000)
    PlotMuSgm(gr)   # requires ggplot2 and grid.
\end{verbatim}
The result is shown in Figure~\ref{fig:mu_sgm_gr}, where all the graphs are presented in the operational time~$\nu$ that corresponds to the total volume of votes:
$$
    t=t(\nu) \df \min\{\tau\in[t_1,t_2]\mid |\mathcal{V}_\tau| \ge \nu\}.
$$
The numbers under the graphs are the volume in thousands of dollars. The price~$q_t^c$ oscillates somewhere between $q_t^-$ and $q_t^+$, while $q^\pm_t$ oscillate themselves throughout 
$[t_1,t_2]$. At the same time, the estimate~$\hat{\bm\mu}_t$ quickly converges and remains extremely stable in spite of huge volumes of offers. 
The estimate $\hat\sigma_t$ measures the average variability of beliefs during $[t_0,t]$. 
In Figure~\ref{fig:mu_sgm_gr}, it begins to decrease after $\hat{\bm\mu}_t$ has stabilized (after the dashed line). 

Suppose $0<\bm\mu(\bm\omega,\bm E) < 1$ and consider a price $q>\bm\mu$. Suppose we are within Assumption~\ref{M_3} (which includes Assumptions~\ref{M_1} and~\ref{M_2}), and suppose we begin 
to reduce~$\sigma_t$. Then, due to Fact~\ref{rem:gap_dec}, after a certain point we will have ${\bm\mu<\theta^+_{\lambda_t}(q)<q,}$
and $q-\theta^+_{\lambda_t}$ will continue to shrink (together with the density of the distribution of beliefs around~$\bm\mu$). Therefore, 
according to 
\eqref{eq:theta_plus}, we will observe an unlimited decrease in the probability of occurrence of beliefs that could generate buy orders for~$q$. Again, 
due to Fact~\ref{rem:gap_dec} and~\eqref{eq:theta_minus}, the same is true for sell orders for any price $q<\bm\mu$. 
Thus, informally speaking, a range of prices where concluded transactions are probable will shrink around~$\bm\mu$. This reasoning, together with 
the constancy of~$\hat{\bm\mu}_t$ and a developing decrease in $\hat\sigma_t$ observed in Figure~\ref{fig:mu_sgm_gr}, leads to the following conclusion.

\begin{Con} Consider the regular market, which is introduced at the beginning of Section~\textup{\ref{subsec:mod_prob}} and satisfies Assumption~\textup{\ref{M_0}}. 
There exists its reasonable model \textup{(}by which we mean Assumption~\textup{\ref{M_3}}\textup{)} 
such that our data, interpreted through this model, do not contradict the following hypothesis. 
In the unlimited market, which is introduced in Assumption~\textup{\ref{asmp:A2}} and is partially
implemented by the regular market, the beliefs $\pi_n^t(\bm\omega)$ would approaches~$\bm\mu$, transactions would conclude at prices increasingly close to $\bm\mu$, and we would finally have 
\begin{equation}\label{eq:final_state_PM2}
    \pi_n^{t_\infty}(\bm\omega)= \bm\mu(\bm\omega,\bm E)=\bm q(\bm\omega,\bm E),\quad n\in\bm E.
\end{equation}
\end{Con}
\begin{Con}\label{con:reg_mkt}
The model under discussion gives an algorithm that allows us to derive, from our data, the estimate~$\hat{\bm\mu}_t$ of $\bm\mu(\bm\omega,\bm E)$ that becomes extremely stable after a certain point. It is reasonable to assume 
that if in the regular market 
the estimate $\hat{\bm\mu}_T$ is public at $T$, 
then we will have
\begin{equation}\label{eq:final_state_M0}
    \pi_n^{T}(\bm\omega)= \hat{\bm\mu}_T,\quad n\in\bm e.
\end{equation}
\end{Con}

Fact~\ref{rem:d_less_N} 
implies that identities~\eqref{eq:final_state_PM2} and~\eqref{eq:final_state_M0} can be considered as an analogue of
relation~\eqref{eq:final_state}. 
We also note that the additional identity $\pi_n^{t_\infty}(\bm\omega)=\bm p(\bm\omega,\bm E)$ does not contradict theoretical and empirical considerations presented in Sections~\ref{subsec:pen_res} and~\ref{subsec:emp_evid}.  
In Section~\ref{subsec:market_dyn}, we show that the data does not contradict the continuous model 
itself, and in 
Sections~\ref{subsec:econ_equil} and~\ref{subsec:util}, we
describe how this model has been chosen. 



\subsection{Interpretable collective intelligence: continuous considerations} 
\label{subsec:sr_reality}
Consider the regular market that is described at the beginning of Section~\ref{subsec:mod_prob} and is assumed to satisfy Assumption~\ref{M_0}. Suppose there is a pre-announced algorithm~$\mathfrak{A}$ for eliciting $\hat{\bm\mu}_{t}$ from $V^{\pm}_t(q)$, and $\hat{\bm\mu}_{T}$ is public at~$T$. 
We introduce the following modifications of the market.
\begin{enumerate}[label = \textup{(md\arabic*)}]
\item\label{md1} 
Each participant $n\in\bm e$
gets one and the same volume $b_n = b$ of \emph{play money} that
can be employed in this market (and only in it): the experts are assumed to have an equal opportunity to affect the market by trading. 
\item\label{md2} 
We 
resolve the market 
in accordance with $\bm\omega[\hat{\bm\mu}_{T}]$.
\end{enumerate}
In addition to Assumption~\ref{M_0}, 
we make an assumption that is a variety of the temporal coherence property.
An expert is temporal coherent if he or she cannot predict whether his or her belief will increase or decrease on average (see also the discussion in~\cite[23--26]{PaIn2009}). In the Bayesian setting, this means that his or her learning scheme forms a martingale. 
In our pre-Ba\-ye\-si\-an language, temporal coherence between~$t$ and~$T$ can be postulated as follows.
\begin{Asmp}[temporal coherence]\label{TC}   
    For $n\in\bm e$ and $t\in[t_0,T]$, we have 
        $$\pi_n^t\big(\bm\omega\big[\pi_n^{T}(\bm\omega)\big]\big) = \pi_n^t(\bm\omega),$$
    where the existence of the left expression is postulated as in Assumption~\ref{asmp:A0} with 
    $\bm\omega\big[\pi_n^{T}(\bm\omega)\big]$ 
    instead of~$\bm\omega$.
    
\end{Asmp}    
We refer to the market in question, which is modified in accordance with~\ref{md1} and~\ref{md2} and satisfies Assumptions~\ref{M_0} and~\ref{TC}, as \emph{the collective intelligence}.
\begin{Def}
The collective intelligence is $\mu$-\emph{efficient} if each expert $n\in\bm e$ satisfies the following conditions.
\begin{enumerate}[label = \textup{($\mu$\arabic*)}]
    \item\label{j1} At each time~$t\in [t_0,T]$, there exists $\delta \in[0, T-t]$ such that the expert~$n$ believes that if he or she 
    publicly generates, during an interval $[t,t+\delta]$, certain signals $\alpha^t_n\assign A^t_n$ 
    (i.e. makes $A^t_n$ a part of $\bigcap_{n\in\bm e} \Gamma_n^{t+\delta}$), then this will lead to the identity 
    $\pi_n^T(\bm\omega) = \hat{\bm\mu}_{T}$ whatever orders the expert~$n$ will place during~$[t,T]$. 
    \item\label{j2} If $n$ places an order~$o_t$, 
    then he or she believes that he or she \emph{will 
    indeed} publicly generate $\alpha^t_n\assign A^t_n$ during $(t,t+\delta]$. 
\end{enumerate}
\end{Def}

We make some comments concerning~\ref{j1} and~\ref{j2}. First, we emphasize that the signals $\alpha^t_n\assign A^t_n$ do not depend on the expert $n$'s planned voting actions: $n$~is going to transmit these signals to others directly. Second, we clarify that $\delta$ depends on $t$ and $n$: 
the interval $[t+\delta,T]$ must be long enough for the collective intelligence to manage to process the new information. Third, we note that Assumption~\ref{TC} implies that if $n$ 
satisfies~\ref{j1} and believes that he or she will generate $\alpha^t_n\assign A^t_n$ during $(t,t+\delta]$, then
\begin{equation}\label{eq:sh_cons_ci}
\pi_n^t(\bm\omega[\hat{\bm\mu}_{T}]) = \pi_n^t\big(\bm\omega\big[\pi_n^{T}(\bm\omega)\big]\big) = \pi_n^t(\bm\omega).
\end{equation}
We come to the following statement. 
\begin{Fact}\label{rem:mu_eff_cor_1}
If an expert satisfies~\ref{j1} and~\ref{j2}, then in the voting he or she relies 
on his or her mental probability of $\bm\omega\mapsto 1$:
when $n\in\bm e$ places an order~$o_t$, 
the probability $\pi_n^t(\bm\omega[\hat{\bm\mu}_{T}])$ exists and coincides with
$\pi_n^t(\bm\omega)$.
\end{Fact}
Finally, substituting $T$ for $t$ in \ref{j1}, we immediately come to the following fact, which may,
due to Fact~\ref{rem:shorten_fs_sr}, 
be viewed as an analogue of 
Fact~\ref{rem:final_state_sr}. 
\begin{Fact}
The property of $\mu$-efficiency includes property~\eqref{eq:final_state_M0}.
\end{Fact}

We present \emph{heuristic} arguments that $\mu$-efficiency is some sort of equilibrium in a sense close 
to Theorem~\ref{lem:eqlb}. 
We note that for further studies, the experimental verification and the practical implementation of our ideas seem to be much more preferable and realistic than a complete formalization of the considerations below. 
We also need to adapt to each other an automated market maker and 
our algorithm for 
eliciting~$\hat{\bm\mu}_t$ (see also Section~\ref{subsec:market_dyn}).

Consider an individual expert $m \in \bm e$ and make the following assumption.
\begin{Asmp}\label{asmp:equil_premise}
    During $[t_0,T]$, the expert $m$ believes
    \begin{enumerate}
        \item that the others experts $n\in \bm{e}\setminus\{m\}$ satisfy
        \ref{j1} and~\ref{j2};
        \item\label{it:alg} that Conclusion~\ref{con:reg_mkt} is true for the original regular market without modifications~\ref{md1} and~\ref{md2},
        but with public~$\hat{\bm\mu}_T$ calculated by~$\mathfrak{A}$.
    \end{enumerate}
\end{Asmp}
Section~\ref{subsec:mod_meets_dat} and further details in Section~\ref{sec:details} justify the existence of a suitable algorithm that induces~\ref{it:alg} and can be taken as~$\mathfrak{A}$.

Due to Fact~\ref{rem:mu_eff_cor_1}, the expert~$m$ believes that the others rely on $\pi_n^t(\bm\omega)$ in their trading. 
From his or her point of view, the collective intelligence looks like a normal market where 
constraints on experts' ability to transmit information through 
trading are compensated by direct information exchange. Concerning the constraints just mentioned, in 
the private setting of the experiment~\cite{AlKiPf2009} (see also Section~\ref{subsec:emp_evid}) we can presumably see that critical information 
may be poorly transmitted through an ordinary market with~\ref{md1}. We may assume that this effect strengthens as the number of experts without this information increases. 
Suppose the expert~$m$ indeed believes he or she cannot distort the others' probabilities by voting. 
Since $m$ regards the collective intelligence as the regular market with an active channel for direct information exchange, and believes that in the regular market the algorithm~$\mathfrak{A}$ would give 
the estimate~$\hat{\bm\mu}_{t}$ remaining stable under the order flaw after a certain point, 
the following assumption is justified.
\begin{Inf}
    At each time $t\in[t_0,T]$, the expert~$m$ believes
    that he or she cannot 
    affect~$\hat{\bm\mu}_{T}$ by his or her voting actions alone. 
\end{Inf}
The expert~$m$ also believes that~$\mathfrak{A}$ gives relation~\eqref{eq:final_state_M0} in the regular market.
Therefore, since $m$ treats the collective intelligence as the regular market with the main information channel separated from the market itself, 
we may assume that if the expert~$m$ timely adds, like the others, his or her information to this channel, 
then he or she will believe that $\pi_m^T(\bm\omega) = \hat{\bm\mu}_{T}$ regardless of his or her voting actions. 
Namely, we can make the following assumption. 
\begin{Inf}\label{inf:mu1_is_true}
    Statement~\ref{j1} is true for~$m$ at the beginning of the voting. It remains true during $[t_0, T]$,
    provided $m$ participates in direct information exchange in accordance with~\ref{j2} so that 
    the significance of the information $\alpha^t_m\assign A^t_m$, intended for sharing, decreases over time (together with the time required for processing this information by the collective intelligence).
    
    
    
\end{Inf}

The main logic leading to~\ref{j2} for~$m$ is contained in the theoretical model 
from Section~\ref{subsec:sr_theory}. 
In the general context under consideration, we also provide certain non\nobreakdash-ri\-go\-ro\-us arguments that     
such behavior is justified.
Let $m$ satisfy~\ref{j1} at time ${t\in [t_0, T]}$, and let $A_m^t\ne\varnothing$. 
Suppose there is a price $q=q_t^+=q_t^-$ that is generated, for example, by an automated market maker and
at which buy and sell orders can be placed at time~$t$.    
This is a deliberate simplification: a real market maker always has a 
bid/ask spread for non-zero orders  
(see also the discussion at the very end of Section~\ref{subsec:market_dyn}). 
Suppose the expert~$m$ considers~$q$  
coherent with the current public information: 
$m$ believes that 
$\hat{\bm\mu}_{T}$ would be equal to 
$q$ 
if he or she remained silent and if, in addition, it turned out that $A_n^t=\varnothing$ for $n\in\bm{e}\setminus\{m\}$.
Assume that 
$$\eta \df {\pi_m^t(\bm\omega)}-{q} 
\ne 0.$$ For definiteness, suppose $\eta>0$.
By virtue of~\eqref{eq:sh_cons_ci}, if $m$ is going to share $\alpha^t_m\assign A^t_m$, then 
\begin{equation}\label{eq:sh_cons_ci_mm}
\pi_m^t(\bm\omega[\hat{\bm\mu}_{T}]) = 
q + \eta.
\end{equation}
By $\hat{\bm\mu}_{T}^{\mathrm{sl}}$ we denote 
the result that $\mathfrak{A}$ will return at $T$ if $m$ remains silent during $[t,T]$. 
Suppose the expert~$m$ anticipates that a part of his or her information is known to others and will be used by them. 
Namely, suppose the following generalization of relation~\eqref{eq:sl_cons} 
holds: 
\begin{equation}\label{eq:sl_cons_ci}
\pi_m^t\bigl(\bm\omega\bigl[\hat{\bm\mu}_{T}^{\mathrm{sl}}\bigr]\bigr) = 
q+ \xi,
\end{equation}
where $\xi \in [0, \eta)$. We may expect that if $m$ did not consider his or her information to be significant or if $m$ believed that other experts know it all, then he or she could set $A_m^t = \varnothing$. This justifies the assumptions $\eta \ne 0$ and $\xi \ne \eta$.

We compare the following \emph{five} alternatives: the expert~$m$ 
\begin{itemize}
    \item does nothing, 
    \item places a buy/sell order at~$q$ and remains silent, 
    \item places
a buy/sell order at~$q$ and publicly generates $\alpha^t_m\assign A^t_m$.
\end{itemize}
Recalling the double relativity principle together with relations~\eqref{eq:Upls} and~\eqref{eq:Umns}, 
and using 
\eqref{eq:sh_cons_ci_mm} and~\eqref{eq:sl_cons_ci}, we can infer that in order to choose the best alternative, we need to select the maximum among the following five values:
$$
1,\quad U^\pm\bigl(q+\xi, q, \lambda_m^t\bigr),\quad\mbox{and}\quad U^\pm\bigl(q+\eta,q,\lambda_m^t\bigr),
$$
where by analogy with Asumption~\ref{M_3}, we suppose $\lambda_m^t < 0$.
Each of these values is not greater than one of the following two of them, with the limits obtained by~\eqref{eq:Upls} and~\eqref{eq:Umns}:
\begin{equation}
\label{eq:limU}
\begin{aligned}
&U^-\bigl(q+\xi, q, \lambda_m^t\bigr)\searrow \psi_q(\xi)\\ 
&U^+\bigl(q+\eta,q,\lambda_m^t\bigr) \searrow \varphi_q(\eta) 
\end{aligned}
\qquad\mbox{as}\quad \bigl|\lambda_m^t\bigr|\searrow 0,
\end{equation}
where $\psi_{q}(\xi) \df 1 - \tfrac{\xi}{1-q} \le 1$ and $\varphi_q(\eta) \df 1 + \tfrac{\eta}{q} > 1$.
In both situations where the significance $|\eta|$ of the information $\alpha^t_m\assign A^t_m$ is sufficiently large or 
the degree~$\bigl|\lambda_m^t\bigr|$ of risk seeking is sufficiently small, the expert~$m$ chooses to vote for $\bm\omega \mapsto 1$ and to share his or her
information. His or her anxiety~$|\xi|$ about others' knowledge enhances the effect.
If $\eta < 0$, then by analogical reasoning, we can draw a similar inference with respect to $\bm\omega \mapsto 0$.

It is interesting to note the following facts. If we interpret $q$ 
as the probability of $\bm\omega[\hat{\bm\mu}_{T}] \mapsto 1$ based on the current public information, then the positive quantity
$$
    \log_2 \varphi_q(\eta) = \log_2 \frac{\pi_m^t(\bm\omega[\hat{\bm\mu}_{T}])}{q} > 0 
$$
can be interpreted exactly as subjective Kharkevich's value of the information 
${\alpha^t_m\assign A^t_m}$ for the expert~$m$ with respect to the goal to make
a profit from a buy order at~$q$. If $\xi > 0$, then the negative quantity
$$
\log_2 \psi_q(\xi) = \log_2 \frac{1-\pi_m^t\bigl(\bm\omega\bigl[\hat{\bm\mu}_{T}^{\mathrm{sl}}\bigr]\bigr)}{1-q} < 0 
$$
is subjective Kharkevich's value of the part of the expert $m$'s information 
known to others,
with respect to the goal to make
a profit from a sell order at~$q$.
The value of information with respect to a certain goal was defined 
by Kharkevich in~\cite{Kh1960}. That definition and its discussion can be found, for example, in the survey~\cite{RoRe2018}.

Further, we still believe that $\bigl|\lambda_m^t\bigr|$ is associated with the expert $m$'s expectations about others' information unknown to him or her, and that this value vanishes together with the dispersion of beliefs (see the discussion following Fact~\ref{rem:gap_dec}). In other words, $\bigl|\lambda_m^t\bigr|$ is associated with the comparative 
ignorance of~$m$ with respect to other experts. In the case of relatively small~$\eta > 0$ and relatively large~$\bigl|\lambda_m^t\bigr|$, it may happen that the value of $U^-$ in~\eqref{eq:limU} 
is greater than the corresponding value of $U^+$. 
On the assumption that $m$ still wants to pursue his or her information advantage, this may imply that $m$ does not consider the present moment suitable for voting at all.  Namely, 
$m$ may choose a sixth alternative: to wait until his or her comparative ignorance, together with $\bigl|\lambda_m^t\bigr|$, decreases. This option does not contradict the above considerations: 
the functions~$U^\pm$ can describe preferences only within the same level~$\lambda_m^t$ of risk seeking,
and the alternatives ``to do nothing'' and ``to wait and act'' are not the same.
If the expert $m$'s comparative ignorance decreases and 
the significance of his or her information does not, then there will come a point when 
placing a buy order together with information sharing will become more attractive than silent voting in the opposite direction. 
The expert~$m$ could find this moment suitable for taking his or her information advantage. 
The waiting time may be long only if the information has little significance and processing it requires little time.
The case $\eta < 0$ can be treated similarly. 
We note that the above considerations are in line with the results of~\cite{FoTv1995}, which show that comparative ignorance may lead to ambiguity aversion behavior. 
Further, we may expect that if $m$ shares his or her information or it leaks out through other experts while $m$ waits, then its significance will decrease and, due to
Inference~\ref{inf:mu1_is_true}, Statement~\ref{j1} will remain true for~$m$. We come to the following inference.

\begin{Inf}
    There is no contradiction in supposing that if the expert~$m$ satisfies Assumption~\ref{asmp:equil_premise}, then \ref{j1} and \ref{j2} will be true for $m$.
\end{Inf}
We make some additional remarks. 
We have suggested that during $[t_0, T]$, the expert~$m$ believes that others believe him or her. 
For this, $m$ have to credibly transmit his or her knowledge. In fact, we may expect that unsupported information or arguments will have little effect 
on~$\hat{\bm\mu}_{t}$. In particular, this implies that opportunities for manipulative behavior are limited. 
Moreover, when a manipulation has an effect on the price, it becomes profitable for everyone (including the manipulator) to reliably negate the manipulative information if possible. If $m$ uses manipulative information to initiate local fluctuations of the price and expects 
that it will eventually be reliably negated, 
then \ref{j1} and \ref{j2} remain true. 
Another reason for non-manipulative behavior and mutual trust among experts could be that, in practice, transmitting manipulative information may involve falsification of scientific research or serious medical consequences.
We also note that there is an extensive literature on manipulation in normal prediction 
markets: see, for example,~\cite{HaOpPo2006,JiSa2012}.



% ...


Further, we note that $m$ may have information that he or she believes to be initially shared among substantial part of the experts in~$\bm e$. For example, in the experiment~\cite{AlKiPf2009} described in Section~\ref{subsec:emp_evid}, such information may consist of signals that develop skills in Bayesian inference. Such background knowledge may not be reflected in $\alpha^t_m\assign A^t_m$, but we may expect that it will be aggregated by $\hat{\bm\mu}_{t}$ (as in the regular market) due to
a significant cumulative impact of the knowledgeable experts.  
In other words, we may expect that the collective intelligence will combine two mechanisms: the principle outlined in Section~\ref{subsec:sr_theory} 
will lead to the exchange of critical information between experts, and
the mechanism~\cite{FeFoPe2005} described in Section~\ref{subsec:pr_mr} 
will be responsible for 
accumulating basic background knowledge. To be more specific, we note that in the $xyz$-example described in Section~\ref{subsec:pr_mr} 
and formalized there as Example~\ref{ex:xyz}, 
each piece of information is shared among half of the experts and each expert initially has some guess of the others' knowledge. This looks rather like an extreme example of background knowledge. On the other hand, the $xyz$-example modified 
as Example~\ref{ex:xyz_sr} 
illustrates how few informed experts can share critical information with numerous ignorant agents.
Finally, we come to the following conclusion.
\begin{Con}\label{con:final}
There is no contradiction in supposing that if we publicly offer the experts to share their information and to vote as though orders will be resolved in accordance with~$\bm\omega$, then
the collective intelligence will be $\mu$-efficient. 
In particular, 
we may expect that~\eqref{eq:final_state_M0} will be true. Here, $\hat{\bm\mu}_T$ may be interpreted as a probability of 
$\bm\omega\mapsto 1$ that is generated by the information shared by the experts during $[t_0,T]$ and combined with their cumulative basic background knowledge that determines how they interpret this information. 
\end{Con}
We have achieved the goal of showing the implementability of the system in question.

\section{Further details}\label{sec:details}
Here we return to the consideration of our data-driven model of conventional CDA markets. 
We discuss important details and suggest directions for future research.
\subsection{Concerning market dynamics}\label{subsec:market_dyn}
We begin with the fact that there are sporting outcomes in~\cite{Data2019} that do not give as good a picture as we see in Figure~\ref{fig:mu_sgm_gr}.
First, after the corresponding estimate~$\hat{\bm\mu}_t$ has stabilized, it may be shifted one more time: the announcement of the starting lineups may break condition~\ref{A2_2} inherited by Assumption~\ref{M_0} from Assumption~\ref{asmp:A2}. 
Second, in some cases we can see a picture that is even more severely distorted. The latter is caused by the fact 
that applying the function~$\log L$ defined by~\eqref{eq:logL}, we rely on an additional implicit assumption that the censoring mechanism in~$\log L$ is non-informative.
This assumption actually requires a further detailed analysis:  
its violation may 
cause the corresponding estimate to be biased (see, e.g., the survey~\cite{LeElAf1997}). 
Below we refine our model by relaxing the assumption of non-informativeness
and introducing a new assumption concerning additional information that can be elicited from experts' decisions. 
We show 
that the refined model is in good agreement with the data~\cite{Data2019}, and use it to explain the market dynamics.

First, we present a relaxed version of the aforementioned implicit assumption that the censoring mechanism $\theta^\pm_{\lambda_t}(q)$ in~\eqref{eq:logL} 
adds no information to inequalities~\eqref{eq:theta_plus} and~\eqref{eq:theta_minus}. Instead, we only assume that 
if we know a price an expert has chosen, but do not know whether the corresponding order is a buy or sell order, then we have no additional
information about the random variable that has generated the expert's belief. 
\begin{Asmp}[a refinement of Assumption~\ref{M_1}]\label{M_1_ext}
Let $t \in [t_0,T]$. Fix an arbitrary price $q\in (0,1)$. 
Let $v$ runs over the votes constituting $V^+_t(q)$ and $V^-_t(q)$. 
We assume that regardless of the choice of~$q$, the corresponding 
beliefs~$\pi_v$ 
are i.i.d. 
with distribution $\mathcal{N}_{[0,1]}(\mu,\sigma_t^2)$, where $\mu$ and $\sigma_t$ are the same as in Assumption~\ref{M_1}.
\end{Asmp}
This assumption can be interpreted as a simplified formalization of the following argument. An expert choose a price largely on the basis of the current market state, and after that 
he or she decides, on the basis of his or her belief, whether the order will be a buy or sell order.

The next assumption formalizes the fact that if we know the price an expert has chosen, then his or her choice between a buy and sell order contains additional information about his or her belief. 
\begin{Asmp}[a refinement of Assumption~\ref{M_2}]\label{M_2_ext}
Let $v$ be one of the votes that has been placed by an expert $n\in\bm e$ at time $\tau \in [t_0,T]$ at some price $q\in (0,1)$.
Consider two possible orders that may have generated $v$: 
$$
o^+_\tau \df \bigl(\tfrac{1}{q}, q, n\bigr)^+\quad\mbox{and}\quad o^-_\tau \df \bigl(\tfrac{1}{1-q}, q, n\bigr)^-.
$$
If $n$ has chosen $o^{\pm}_\tau$, then 
\begin{equation}\label{eq:conv_comb}
\rho\bigl(o_\tau^\pm\bigr)\, U^\pm(\pi_v, q, \lambda_n^\tau) + \bigl(1-\rho\bigl(o_\tau^\pm\bigr)\bigr) \ge \rho\bigl(o_\tau^\mp\bigr)\, U^\mp(\pi_v, q, \lambda_n^\tau) + \bigl(1-\rho\bigl(o_\tau^\mp\bigr)\bigr)
\end{equation}
for certain values $\rho\bigl(o_\tau^\pm\bigr) \in [0,1]$.
\end{Asmp}
The meaning of this assumption becomes clear if we recall that 
\begin{itemize}
    \item according to Assumption~\ref{M_1}, each vote has been generated by a separate order;
    \item according to Assumption~\ref{M_2}, each expert believes that each of his or her new orders will be either completely matched or completely unmatched; 
    \item we always have $u(b \mid \lambda, b, L)=1$; 
    \item $U^\pm$ are functions for calculating the expected utility of fully matched orders.
\end{itemize}

% ...

If just before $\tau$ we have $S^{\pm}_{\tau-}(q) < S^{\mp}_{\tau-}(q)$, then we should set $\rho\bigl(o_\tau^\pm\bigr) =1$ (because our assumptions imply that if a new order is at least partially matched, then the expert who has placed it has no option but to believe that it will be fully matched). 
In other cases, $\rho\bigl(o_\tau^\pm\bigr)$ may depend on the difference $S^{+}_{\tau-}(q) - S^{-}_{\tau-}(q)$, 
on the remaining time $T-\tau$, on the price~$q$, etc.
We simplify the situation applying the same averaging approach that we have employed to introduce~$\sigma_t$ and~$\lambda_t$: we spread, over each interval $[t_0, t]$, certain invariable 
probabilities of unmatched votes being matched.
\begin{Asmp}[averaging of $\rho$]\label{M_4}
Fix $t \in [t_0,T]$. If $v \in \mathcal{V}_t$, then the corresponding matching probabilities in Assumption~\ref{M_2_ext} can be calculated as follows:  
\begin{equation}\label{eq:rho}
\rho\bigl(o_\tau^\pm\bigr) = 
\begin{cases}
\rhofn_t^\pm(q),& S^{\pm}_{\tau-}(q) \ge S^{\mp}_{\tau-}(q);\\
1,& S^{\pm}_{\tau-}(q) < S^{\mp}_{\tau-}(q),
\end{cases}
\end{equation}
where $\rhofn_t^+$ and $\rhofn_t^-$ are respectively increasing and decreasing functions from~$(0,1)$ to~$(0,1)$.
\end{Asmp}

Suppose we are under the described model with given parameters $\mu$, $\sigma_t$, and $\lambda_t = \lmd(\mu,\sigma_t)$. 
Fix a price $q\in(0,1)$ and suppose we have, just before $\tau \in [t_0,t]$, the corresponding volumes~$V^\pm_{\tau-}(q)$.
Suppose we are going to add to them a new vote~$v$ with $\pi_v$ generated by $\mathcal{N}_{[0,1]}(\mu,\sigma_t^2)$.
If $\pi_v \ge \theta^-_{\lambda_t}(q)$, then we have $$U^-(\pi_v, q, \lambda_t) \le 1 < U^+(\pi_v, q, \lambda_t)$$ 
(see Proposition~\ref{lem:cenz_mech}) and the same is true for the convex combinations of these expressions and~$1$ in inequalities~\eqref{eq:conv_comb}.
Thus, we have to add the vote~$v$ to the votes for $\bm\omega \mapsto 1$ and set ${V^+_\tau(q) = V^+_{\tau-}(q) + 1}$. 
Similarly, if 
$\pi_v\le \theta^+_{\lambda_t}(q)$, 
then we have to add~$v$ to the votes for~$\bm\omega \mapsto 0$
and set ${V^-_\tau(q) = V^-_{\tau-}(q) + 1}$.
Next, if 
\begin{equation}\label{eq:eq_zn}
\pi_v \in \big(\theta^+_{\lambda_t}(q), \theta^-_{\lambda_t}(q)\big),
\end{equation} 
then both $U^\pm(\pi_v, q, \lambda_t)$ are greater than~$1$ and according 
to~\eqref{eq:conv_comb} and~\eqref{eq:rho}
there is an incentive to use $v$ for maintaining the balance between supply and demand. For example, if $S^{\pm}_{\tau-}(q) < S^{\mp}_{\tau-}(q)$ and $\rhofn_t^\mp(q)$ is sufficiently small, then we have to increase~$V^\pm_\tau(q)$.

We refer to intervals~\eqref{eq:eq_zn} as 
\emph{equilibrium zones}.
If $q$ is not too far from~$\mu$ and the probability measure $F_{\mu,\sigma_t}\big(\theta^-_{\lambda_t}(q)\big) - F_{\mu,\sigma_t}\big(\theta^+_{\lambda_t}(q)\big)$ of the corresponding equilibrium zone is sufficiently large, then a situation will arise 
where the Markov chain ${S^{+}_{\tau}(q) - S^{-}_{\tau}(q)}$ does not stray too far from zero and repeatedly vanishes as 
the operational time ${V^{+}_{\tau}(q) + V^{-}_{\tau}(q)}$ increases. A mathematically rigorous description and proof of this effect may be a subject of future research.
Here we demonstrate it with an example by using stochastic simulation. 
In the following calls made in \path{CI.RData} from~\cite{Data2019}, we set $\mu = 0.261$, $\sigma_t = 0.003$, and place $5000$ votes at the price $q = 0.247$ with 
$\rhofn_t^+(q) = 0.01$ and $\rhofn_t^-(q) = 0.999$ in accordance with our model. 
\begin{samepage}
\begin{verbatim}
    sdHist <- GetDemSuppForq(q = 0.247, mu = 0.261, sgm = 0.003, 
                             rhoPls = 0.01, rhoMns = 0.999, 
                             nVotes = 5000, seed = 2022)
    PlotDemSuppDiff(sdHist)   # requires ggplot2
\end{verbatim}
\end{samepage}
\begin{figure}
\includegraphics[scale = 0.6]{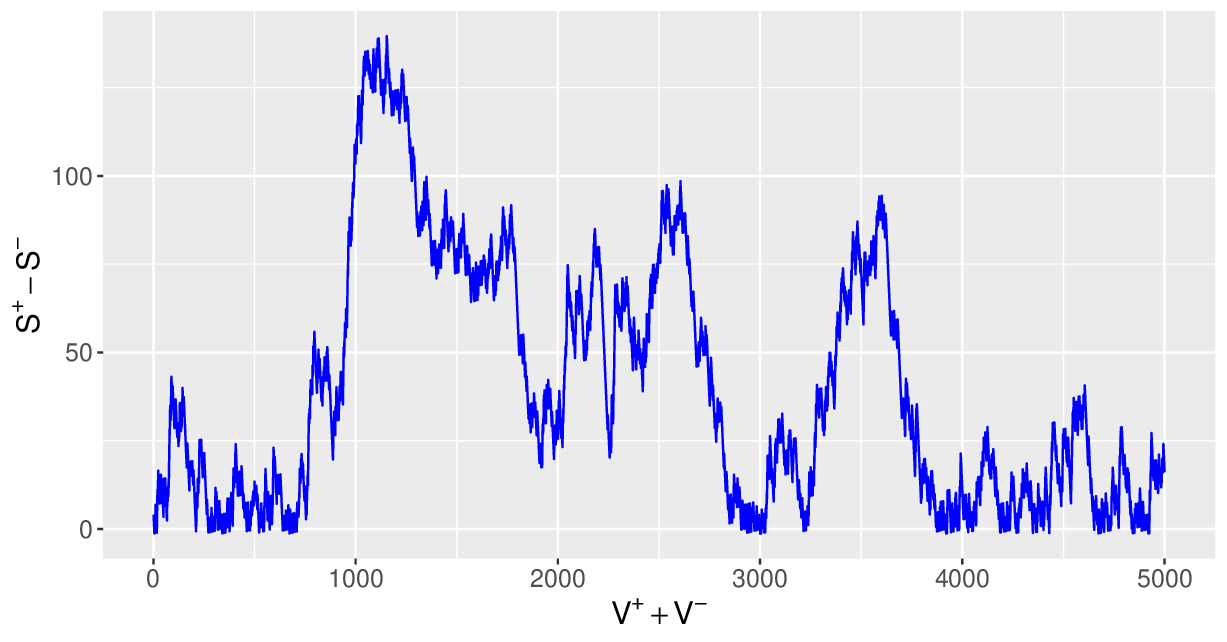}
\caption{Difference between demand and supply.}
\label{fig:supp_dem_diff}
\end{figure}
The result is presented in Figure~\ref{fig:supp_dem_diff}: the non-equilibrium state tends to occur, but cannot persist for long.
This case is borderline: if we slightly increase $\mu-q$ or $\rhofn_t^+(q)$, then the observed tendency toward equilibrium will be immediately broken. On the other hand, the closer~$q$ is 
to~$\mu$, the more freedom we have in choosing~$\rhofn_t^\pm(q)$ so that the tendency toward equilibrium remains.

Thus, we can model a range of prices around~$\mu$ where the differences between demand and supply repeatedly vanish. 
In light of this fact, it is interesting to discuss the dynamics of prices in real CDA mar\-kets. 
We recall that $Q$ is a grid of prices available for trading in such a market. The varying set $[q_\tau^-,q_\tau^+]\cap Q$ consists exactly of the prices that can accept both buy and sell orders, and traders use these prices for offers in the first place. 
When one of the boundary inequalities $S^{\pm}_{\tau}(q_\tau^\pm) < S^{\mp}_{\tau}(q_\tau^\pm)$ turns into equality, the interval $[q_\tau^-,q_\tau^+]$ extends in the corresponding direction. Conversely, this interval narrows when an equality $S^{+}_{\tau}(q) = S^{-}_{\tau}(q)$ breaks for some 
price $q \in (q_\tau^-,q_\tau^+)$, which becomes its new boundary. Thus, the intense dynamics of $q_\tau^\pm$ can be explained by the presence of a range of prices with the above-described tendency toward equilibrium. 

Using stochastic simulation, we now demonstrate that our model can reproduce real data with reasonable accuracy. 
Again, more detailed analysis may be a subject of further research.
We consider the table \verb+data12+ preloaded into \path{CI.RData}, which contains the trading history for a certain 
soccer event. 
Suppose $\bm\omega$ represents
the draw outcome of this event and $t$ is the first moment when the corresponding total volume~${\sum_{q\in Q}\bigl(V^+_t(q) + V^-_t(q)\bigr)}$ 
becomes greater than $15000$ dollars. We choose certain parameters $\mu$ and $\sigma_t$ together with certain functions $\rhofn^\pm_t$.
For each ${q\in Q}$, we generate $\lfloor V^+_t(q) + V^-_t(q) \rfloor$ votes in accordance with our model with the chosen parameters 
and compare the results with the real volumes $V^\pm_t(q)$. 
\begin{figure}
\includegraphics[scale = 0.6]{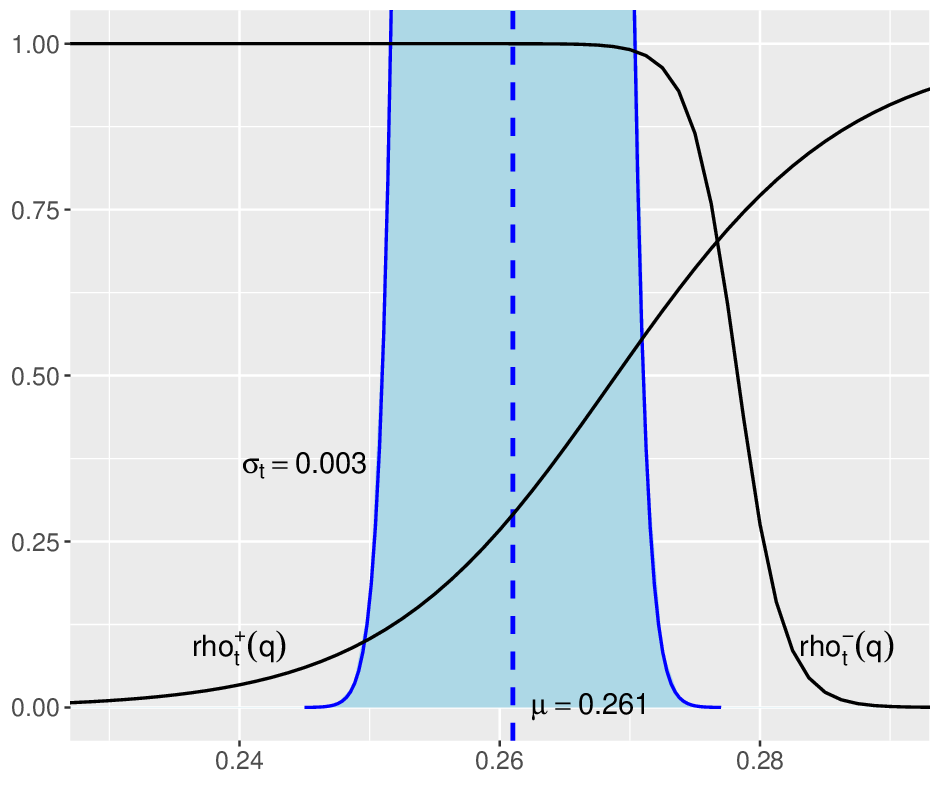}
\caption{Model parameters.}
\label{fig:mu_sgm_rho_12_3}
\end{figure}
The chosen parameters 
are visualized in Figure~\ref{fig:mu_sgm_rho_12_3}: we see a part of the corresponding density function as well as the functions $\rhofn^\pm_t$ that have the form
$$
\rhofn^\pm_t(q) = \expit\bigl(k_t^{\pm}\bigl(\logit(q) - d_t^\pm\bigr)\bigr),\quad k_t^{\pm},d_t^\pm \in \mathbb{R},
$$
where 
$\logit\colon (0,1) \to \mathbb{R}$
and
$\expit\colon \mathbb{R} \to (0,1)$
are the 
standard logit and sigmoid functions defined by the formulas
$$
\logit(q) \df \log\frac{q}{1-q}\quad\mbox{and}\quad \expit(x) \df \frac{1}{1+e^{-x}}.
$$
Concerning the results of generating data, they can be obtained by making the call
\begin{verbatim}
    # reproduces first 15000 votes (dollars) at the draw outcome in
    # UEFA Champions League, 23 July 2013, Hafnarfjordur vs. Ekranas
    ModelExample()
\end{verbatim}
from \path{CI.RData} and
the main body 
of them is presented in Table~\ref{tbl:model_ex}: we see a similarity between the real and generated data, which dramatically increases when we 
smooth out fluctuations by summing some volumes over neighboring prices. 
\begin{table}
\caption{Results of generating data.}
\label{tbl:model_ex}
\begin{tabular}{@{}ccccccccc@{}}
\toprule[1.0pt]
       & \multicolumn{4}{c}{Real data}                                                                               & \multicolumn{4}{c}{Model data}                                                                      \\
$q$    & \multicolumn{2}{c}{$V^+$}                            & \multicolumn{2}{c}{$V^-$}                            & \multicolumn{2}{c}{$V^+$}                        & \multicolumn{2}{c}{$V^-$}                        \\ \cmidrule(r){1-1}\cmidrule(lr){2-5}\cmidrule(l){6-9}
0.2222 & \multicolumn{2}{c}{13.4}                             & \multicolumn{2}{c}{0.9}                              & \multicolumn{2}{c}{14}                           & \multicolumn{2}{c}{0}                            \\
0.2326 & \multicolumn{2}{c}{243.7}                            & \multicolumn{2}{c}{0.0}                              & \multicolumn{2}{c}{243}                          & \multicolumn{2}{c}{0}                            \\
0.2381 & \multicolumn{2}{c}{50.1}                             & \multicolumn{2}{c}{0.0}                              & \multicolumn{2}{c}{50}                           & \multicolumn{2}{c}{0}                            \\
0.2439 & \multicolumn{2}{c}{282.2}                            & \multicolumn{2}{c}{0.0}                              & \multicolumn{2}{c}{268}                          & \multicolumn{2}{c}{14}                           \\ \cmidrule(lr){2-3}\cmidrule(lr){4-5}\cmidrule(lr){6-7}\cmidrule(l){8-9} 
0.2500 & 571.8  & \multirow{2}{*}{\makecell{$\sum=$\\1146.0}}                     & 280.3  & \multirow{2}{*}{\makecell{$\sum=$\\541.1}}                      & 638  & \multirow{2}{*}{\makecell{$\sum=$\\1142}}                     & 214  & \multirow{2}{*}{\makecell{$\sum=$\\545}}                      \\
0.2532 & 574.2  &                                             & 260.9  &                                             & 504  &                                           & 331  &                                           \\ \cmidrule(lr){2-3}\cmidrule(lr){4-5}\cmidrule(lr){6-7}\cmidrule(l){8-9} 
0.2564 & 1075.9 & \multicolumn{1}{l}{\multirow{4}{*}{\makecell{$\sum=$\\3385.1}}} & 1065.3 & \multicolumn{1}{l}{\multirow{4}{*}{\makecell{$\sum=$\\ 6733.5}}} & 1013 & \multicolumn{1}{l}{\multirow{4}{*}{\makecell{$\sum=$\\3407}}} & 1128 & \multicolumn{1}{l}{\multirow{4}{*}{\makecell{$\sum=$\\6710}}} \\
0.2597 & 1236.4 & \multicolumn{1}{l}{}                        & 2389.4 & \multicolumn{1}{l}{}                        & 1240 & \multicolumn{1}{l}{}                      & 2385 & \multicolumn{1}{l}{}                      \\
0.2632 & 801.4  & \multicolumn{1}{l}{}                        & 2300.0 & \multicolumn{1}{l}{}                        & 819  & \multicolumn{1}{l}{}                      & 2282 & \multicolumn{1}{l}{}                      \\
0.2667 & 271.5  & \multicolumn{1}{l}{}                        & 978.7  & \multicolumn{1}{l}{}                        & 335  & \multicolumn{1}{l}{}                      & 915  & \multicolumn{1}{l}{}                      \\ \cmidrule(lr){2-3}\cmidrule(lr){4-5}\cmidrule(lr){6-7}\cmidrule(l){8-9} 
0.2703 & 75.7   & \multirow{4}{*}{\makecell{$\sum=$\\245.6}}                      & 434.3  & \multirow{4}{*}{\makecell{$\sum=$\\1903.6}}                     & 137  & \multirow{4}{*}{\makecell{$\sum=$\\247}}                      & 373  & \multirow{4}{*}{\makecell{$\sum=$\\1901}}                     \\
0.2740 & 86.6   &                                             & 453.0  &                                             & 43   &                                           & 496  &                                           \\
0.2778 & 0.0    &                                             & 622.5  &                                             & 25   &                                           & 597  &                                           \\
0.2817 & 83.4   &                                             & 393.7  &                                             & 42   &                                           & 435  &                                           \\ \cmidrule(lr){2-3}\cmidrule(lr){4-5}\cmidrule(lr){6-7}\cmidrule(l){8-9} 
0.2857 & \multicolumn{2}{c}{0.0}                              & \multicolumn{2}{c}{13.2}                             & \multicolumn{2}{c}{0.0}                          & \multicolumn{2}{c}{13}                           \\
0.2985 & \multicolumn{2}{c}{0.0}                              & \multicolumn{2}{c}{360.8}                            & \multicolumn{2}{c}{0.0}                          & \multicolumn{2}{c}{360}                          \\
0.3125 & \multicolumn{2}{c}{0.0}                              & \multicolumn{2}{c}{337.8}                            & \multicolumn{2}{c}{0.0}                          & \multicolumn{2}{c}{337}                          \\
0.3226 & \multicolumn{2}{c}{2.5}                              & \multicolumn{2}{c}{5.3}                              & \multicolumn{2}{c}{0.0}                          & \multicolumn{2}{c}{7}                            \\
0.3279 & \multicolumn{2}{c}{0.0}                              & \multicolumn{2}{c}{69.1}                             & \multicolumn{2}{c}{0.0}                          & \multicolumn{2}{c}{69}                           \\ \bottomrule[1.0pt]
\end{tabular}
\end{table}

We note that the parameters presented in Figure~\ref{fig:mu_sgm_rho_12_3} have been adjusted manually. The algorithm we apply in Section~\ref{subsec:mod_meets_dat} and
describe in Section~\ref{subsec:econ_equil} does not take the additional model assumptions into account, and this may have a significant impact on the results, provided transaction prices are far from~$\tfrac{1}{2}$. In future research we need to adapt the algorithm in accordance with the extended model. This may be easier in the setting where prices are provided by an automated market maker and matching probabilities equal~$1$. 
More specifically, consider the self-re\-solv\-ing scenario 
and suppose there is an automated market maker that provides, at each time $t\in [t_0,T]$, two prices $q^+_t$ and $q^-_t$ for placing, with immediate matching, one vote for or against $\bm\omega \mapsto 1$, respectively.
We may expect that the greater an expert's information advantage is, the sooner taking this advantage will become more preferable for him or her than waiting for 
the reduction of his or her comparative ignorance. This means that the significance of the information being shared will decrease with time.
Thus, we may expect $\pi_v$ to exhibit damped oscillations 
in the manner similar to what we have observed in conventional prediction markets. This leads us to believe that we can rely on almost the same model assumptions. 
Namely, we can 
stay with the normality assumption~\ref{M_1} 
and describe, in the same manner as in Assumptions~\ref{M_2} and~\ref{M_2_ext}, the choice between
\begin{itemize}
    \item inaction,
    \item placing one vote for $\bm\omega \mapsto 1$ at $q^+_t$,
    \item placing one vote for $\bm\omega \mapsto 0$ at $q^-_t$,
\end{itemize}
while both matching probabilities 
are equal to~$1$. 
It is also possible to adapt the function $\log L$ accordingly. 
Concerning the market maker implementation, we refer to the article~\cite{OtSa2012} which describes a customizable and dynamically controllable market maker. 
We believe that we can customize such a market maker so that the estimate~$\hat{\bm\mu}_{t}$ will quickly converge to the parameter~$\mu$ within our model
(in particular, it will be possible to verify this by using stochastic simulation). This will lead us to the implementation of the collective intelligence described in Section~\ref{subsec:sr_reality}.

\subsection{Concerning economic equilibrium}\label{subsec:econ_equil} 
In the papers~\cite{Ma2006,Gje2004,WoZi2006,OtSo2007}, the authors compare the equilibrium price and a central tendency of the distribution 
of traders' beliefs in a prediction market. We also provide an equation that postulates that~$\mu$ is the equilibrium price in a certain strong sense.
This equation produces the function~$\lmd(\cdot,\cdot)$ being utilized in all our calculations. 

In the previous section, we see that our model gives a range of equilibrium prices where the differences ${S^{+}_{\tau}(q) - S^{-}_{\tau}(q)}$ exhibit bounded oscillations near~$0$, or, what is the same, the ratios
${S^{+}_{\tau}(q)}/{S^{-}_{\tau}(q)}$ exhibit damped oscillations near~$1$. The latter may be interpreted to mean that for each such $q$ the frequencies of 
votes for $\bm\omega\mapsto 1$ and votes for $\bm\omega\mapsto 0$ are in the ratio of $q$ to $1-q$. We now choose the ``most equilibrium'' price and postulate that it coincides with~$\mu$. Namely, we write an equation that can be interpreted as follows. For 
the price~$\mu$, 
the frequencies of votes that can \emph{only} be for $\bm\omega\mapsto 1$,
and votes that can \emph{only} be for $\bm\omega\mapsto 0$,
have to be in the ratio of $\mu$ to $1-\mu$.
In accordance with Assumptions~\ref{M_2}, \ref{M_2_ext}, and Proposition~\ref{lem:cenz_mech}, $v$~cannot be placed at a price~$q$ as a vote for $\bm\omega\mapsto 1$ iff 
$\pi_v \le \theta^+_{\lambda_n^\tau}(q)$, and it cannot be vote for $\bm\omega\mapsto 0$ 
iff $\pi_v \ge \theta^-_{\lambda_n^\tau}(q)$.
In the following proposition, we assume the existence of a common parameter~$\lambda < 0$, 
write our strong equilibrium equation, and derive 
the function $\lmd(\cdot,\cdot)$ from it.
\begin{Le}
Consider the function
$$
R(\lambda \mid \mu, \sigma) \df 
\frac{1-\mu}{\mu}\,
\frac{1-F_{\mu,\sigma}\bigl(\theta^-_{\lambda}(\mu)\bigr)}{F_{\mu,\sigma}\bigl(\theta^+_{\lambda}(\mu)\bigr)},
$$
where $\lambda < 0$, $0<\mu<1$, and $\sigma > 0$.
    Fix $\mu \ne \frac{1}{2}$. There exists $d_\mu > 0$ such that for any fixed $\sigma \in (0, d_\mu)$, the equation 
        \begin{equation}\label{eq:log_R}
        \log R(\lambda\mid \mu, \sigma) = 0
        \end{equation}
        has a unique negative solution, which we denote by $\lmd(\mu,\sigma)$. It satisfies~\eqref{eq:lmd_prop}.
\end{Le}
\begin{proof}
We suppose $\mu<\tfrac{1}{2}$. The considerations for $\mu>\tfrac{1}{2}$ are symmetrical.
Denote $G_{\mu,\sigma}(x)\df \frac{f_{\mu,\sigma}(x)}{F_{\mu,\sigma}(x)}$. 
The proofs of the following three relations are simple and we omit them:
    \begin{gather}\label{eq:d_G}
    \partial_x G_{\mu,\sigma}(x)<0;\\[8pt]\label{eq:d_theta}
    |\partial_\lambda\theta^+_{\lambda}(\mu)|<|\partial_\lambda\theta^-_{\lambda}(\mu)|\quad\mbox{for}\quad\mu<\tfrac{1}{2};
    \\[8pt]\label{eq:lim_W}
    a<\mu< b,\quad b-\mu>\mu-a \quad\Longrightarrow\quad
    \frac{1-F_{\mu,\sigma}(b)}{F_{\mu,\sigma}(a)} \searrow 0\quad\mbox{as}\quad \sigma \searrow 0.
    \end{gather}
Since
$$
    R(0\mid \mu, \sigma) = \frac{1-\mu}{\mu} > 1\quad\mbox{and}\quad 
    \lim_{\lambda \to -\infty} R(\lambda\mid \mu, \sigma) = 
    \frac{1-\mu}{\mu}\,\frac{1-F_{\mu,\sigma}(1)}{F_{\mu,\sigma}(0)}, 
$$
relation~\eqref{eq:lim_W} implies that equation~\eqref{eq:log_R} 
has at least one solution, provided $\sigma$ is not too large.
Relation~\eqref{eq:d_theta} implies that 
\begin{equation}\label{eq:tht_mu}
     {\theta}^-_{\lambda}(\mu) -\mu >  \mu - \theta^+_{\lambda}(\mu).
\end{equation}
Differentiating~$\log R$ and combining~\eqref{eq:d_G}, \eqref{eq:d_theta}, and~\eqref{eq:tht_mu}, we obtain
$$
\partial_\lambda \log R(\lambda\mid\mu,\sigma) 
= -\partial_\lambda\theta^+_{\lambda}(\mu)\,G_{\mu,\sigma}\bigl(\theta^+_{\lambda}(\mu)\bigr) -
\partial_\lambda\theta^-_{\lambda}(\mu)\,G_{\mu,\sigma}\bigl(2\mu - {\theta}^-_{\lambda}(\mu)\bigr) > 0.
$$
Thus, a solution of~\eqref{eq:log_R} is unique.
Finally, \eqref{eq:lim_W} and~\eqref{eq:tht_mu} imply~\eqref{eq:lmd_prop}. 
\end{proof}
Analyzing the functions $\log R$ and $\lmd(\cdot,\cdot)$ numerically and graphically (without rigorous analytical proofs), we may additionally infer
that
\begin{itemize}
    \item the closer $\mu$ is to $\frac{1}{2}$, the greater $d_\mu$ can be chosen;
    \item at least we have $d_\mu> \min(\mu,1-\mu)$;
    \item there exists a finite limit
    $$
        \lim_{\mu\to\frac{1}{2}} \lmd(\mu,\sigma) < 0,
    $$
    which can be taken as $\lmd\big(\tfrac{1}{2},\sigma\big)$.
\end{itemize}

We also note that in Assumption~\ref{M_2} we describe only how the experts compare trading actions with the inaction, and in fact it suffices in order for equation~\eqref{eq:log_R} to be meaningful. In Assumption~\ref{M_2_ext}, we additionally postulate how the experts compare buying and selling, and this
more detailed (but more restrictive) model allows us to consider an alternative equilibrium equation that can be interpreted as follows. If the price~$\mu$ accepts all orders with immediate matching, then the frequencies of votes for $\bm\omega\mapsto 1$ and votes for $\bm\omega\mapsto 0$ will be in the ratio of $\mu$ to $1-\mu$. 
Namely, we can consider the equation
\begin{equation}\label{eq:alt_eq_eq}
\frac{1-\mu}{\mu}\,
\frac{1-F_{\mu,\sigma}(\pi_{\mu,\lambda})}{F_{\mu,\sigma}(\pi_{\mu,\lambda})} = 1,
\end{equation}
where $\pi_{\mu,\lambda}$ is a unique solution of the equation
$$
U^+(\pi,\mu,\lambda) = U^-(\pi,\mu,\lambda).
$$
The analysis of equation~\eqref{eq:alt_eq_eq} is another subject of future research.

Now we discuss the details of the algorithm we apply in Section~\ref{subsec:mod_meets_dat}.
We define
\begin{equation}\label{eq:arg_max_logL}
    \bigl(\hat\mu(\lambda), \hat\sigma(\lambda)\bigr) \df \argmax_{(\mu,\sigma)} \log L(\mu,\sigma \mid\lambda)
\end{equation}
and consider the equation  
\begin{equation}\label{eq:outer_eq}
\lambda = \lmd\bigl(\hat\mu(\lambda), \hat\sigma(\lambda)\bigr),\quad\lambda < 0.
\end{equation}
We can observe that when we increase $|\lambda|$, the value $\hat\sigma(\lambda)$ decreases together with the absolute value of the right side of~\eqref{eq:outer_eq}. 
Thus, the left and right sides of equation~\eqref{eq:outer_eq} vary in different directions and its solution~$\hat\lambda_t$ is a balance point where they coincide.
Finally, we set
$$
\hat{\bm\mu}_t = \hat\mu\bigl(\hat\lambda_t\bigr)\quad\mbox{and}\quad\hat{\sigma}_t = \hat\sigma\bigl(\hat\lambda_t\bigr).
$$
In order to compute $\lmd(\mu,\sigma)$ as a solution of~\eqref{eq:log_R}, 
to compute~\eqref{eq:arg_max_logL}, and to solve~\eqref{eq:outer_eq}, we use Nelder--Mead method~\cite{NeMe1965} of optimization. 

\subsection{Concerning utilities}\label{subsec:util} 
We now describe how the double relativity principle has been deduced from the data.
This provides further justification for our choice of utility. 

First, we have begun with the assumption that
each expert in $\bm e$ places 
$b_n$ votes by means of a single order, where 
$b_n = b$ are the same for all the experts, and without loss of generality we have set 
$b=1$.
Second, we have assumed that having been completely matched, each such order is at least better than inaction in terms of the standard exponential utility 
\begin{equation}\label{eq:exp_u}
u(x\mid\alpha) = -\frac{e^{-\alpha x}}{\alpha}.
\end{equation}
Third, we have assumed that the voters for $\bm\omega\mapsto 1$ have one and the same risk attitude $\alpha = \alpha^+$, 
while the voters for $\bm\omega\mapsto 0$ have $\alpha = \alpha^-$.

Under those preliminary assumptions, we have considered our data for 
various sporting outcomes~$\bm\omega$.
We have deduced that 
for each $t$ from the corresponding observation 
interval $[t_1, t_2]$ and for each $q\in Q$, we have $V^+_t(q)$ inequalities
\begin{equation}\label{eq:pi_exp_pls}
\pi_v > \frac{1-\exp\bigl({-\alpha^+}\bigr)}{1 - \exp\Bigl({-\tfrac{\alpha^+}{q}}\Bigr)}
\end{equation}
and $V^-_t(q)$ inequalities
\begin{equation}\label{eq:pi_exp_mns}
\pi_v < \frac{1 - \exp\frac{\alpha^- q}{1-q}}{1 - \exp\frac{\alpha^-}{1-q}}.
\end{equation}
Assuming the normality of~$\pi_v$ and applying the log-likelihood function with the censoring mechanism given by~\eqref{eq:pi_exp_pls} and~\eqref{eq:pi_exp_mns}, we have obtained
estimates $\hat\mu_t(\alpha^+,\alpha^-)$ related to various~$\bm\omega$. 
We have observed that in order for $\hat\mu_{t_2}(\alpha^+,\alpha^-)$ to look like the central value of fluctuations of $q^c_t$, we need to 
choose $\alpha^\pm$ such that $\alpha^\pm < 0$ and 
$\frac{\alpha^+}{\alpha^-}\approx \frac{q^c_{t_2}}{1-q^c_{t_2}}$.
But the prices $q \in Q$ with major volumes $V^\pm_t(q)$ are similar to each other, and $q^c_{t_2}$ is among them and
may be considered as their rough average. This has prompted us to substitute 
\begin{equation}\label{eq:alpha_subst}
\alpha^+ = \lambda q\quad\mbox{and}\quad\alpha^- = \lambda (1-q),
\end{equation} 
where $\lambda$ is a negative value. We note at once that such substitution transforms the expressions 
in~\eqref{eq:pi_exp_pls} and~\eqref{eq:pi_exp_mns} to $\theta^+_{\lambda}(q)$ and $\theta^-_{\lambda}(q)$, respectively.

We can explicitly derive the double relativity principle from~\eqref{eq:alpha_subst}. Let $\mathcal{L}$ be the set of all lotteries. 
For a lottery~$L$ generated by a matched vote, we have 
either $M_L-m_L = \tfrac{1}{q}$ if the vote is for $\bm\omega \mapsto 1$, or 
$M_L-m_L = \tfrac{1}{1-q}$ if the vote is for $\bm\omega \mapsto 0$.
Taking this into account, we can see that comparing a matched order with the inaction by means of function~\eqref{eq:exp_u} with parameters~\eqref{eq:alpha_subst} fits into a scheme where
an individual with a budget~$b$ compares lotteries 
\emph{within} each of the sets
$$
    \mathcal{L}(m, M) = \{L \in \mathcal{L} \mid m_L = m,\;M_L = M\} \cup \{L_b\},\quad m<b<M,
$$
by means of the corresponding utility function
$$
    u(x\mid \lambda, m, M) = -\frac{\exp \bigl(-\lambda\frac{x-m}{M-m}\bigr)}{\lambda},
	\quad \lambda \in \mathbb{R},\quad x\in [m, M].
$$
Thus, we have derived the first relativity from the data. This suffices to describe how an individual compares lotteries with his or her inaction~$L_b$, and to obtain the formulas for $\theta_\lambda^\pm(q)$ together with the log-likelihood function~\eqref{eq:logL} and equation~\eqref{eq:log_R}.

We now note that the inaction~$L_b$ is a single element of the intersection of the sets $\mathcal{L}(m, M)$, $m<b<M$, and its expected utility takes different values
$u(b\mid \lambda, m, M)$ in different sets $\mathcal{L}(m, M)$. In order to solve this problem, we have applied affine transformations to the functions $u(\cdot\mid \lambda, m, M)$ so that they 
become equal at~$b$ and the order within each set $\mathcal{L}(m, M)$ remains the same. This has led us to the second relativity and to the unified utility function~\eqref{eq:drp}, 
which
is increasing for any $\lambda\in\mathbb{R}$ and describes how an individual compares any two lotteries in
$$
\mathcal{L}_b \df\!\!\! \bigcup_{m<b<M}\!\!\! \mathcal{L}(m, M).
$$
These lotteries may be from different sets $\mathcal{L}(m, M)$ and, in particular, we have got the opportunity 
to formulate the refined assumption~\ref{M_2_ext} (together with equation~\eqref{eq:alt_eq_eq} which needs further investigation).
Thus, the second relativity, which has a theoretical origin, is indirectly supported 
by the results of Section~\ref{subsec:market_dyn} (see Table~\ref{tbl:model_ex}). 

Finally, we note that function $u(x\mid \lambda, b, L)$ 
satisfies the von Neu\-mann--Mor\-gen\-stern axioms within 
each of the convex sets $\mathcal{L}(m, M)$, $m<b<M$. A complete axiomatization of the double relativity principle may be a subject of future research.
In this regard, we note that the theory built in~\cite{Co1992} has certain similarities with our constructions: it takes into account \emph{security factors}~$m_L$ and \emph{potential
factors}~$M_L$ of lotteries $L\in\mathcal{L}$. Nevertheless, our function $u(x\mid \lambda, b, L)$ 
does not fit into the axioms of~\cite{Co1992}, which, in particular, do not highlight the role of the status quo lottery~$L_b$. On the other hand, the notion of \emph{the status quo bias} is also studied in the literature
(see, e.g., \cite{SaZe1988,YuOk2005}). It may be interesting to establish how our constructions are connected with those studies as well as with the prospect theory~\cite{KaTv1979}.

\end{document}